\documentclass[onefignum,onetabnum]{siamart171218}
\usepackage{algorithmic}
\usepackage{algorithm}
\usepackage{bm}
\makeatletter
\renewcommand{\ALG@name}{\sc Algorithm}
\newsiamremark{example}{\sc Example}
\newsiamremark{remark}{\sc Remark}

\title{Sum of Squares Decompositions of Polynomials over their Gradient Ideals
	with Rational Coefficients}
\author{Victor Magron, Mohab Safey El Din, and  Trung-Hieu Vu}
\headers{SOS Decomposition of Polynomials over their Gradient Ideals}{Magron,
	Safey El Din, and Vu}
\ifpdf
\hypersetup{pdftitle={SOS Decomposition},pdfauthor={MSH}}
\usepackage{latexsym,amsmath,amsfonts,amscd,amssymb,verbatim,array,mathdots}

\usepackage{chngcntr}
\counterwithin{equation}{section}
\DeclareMathOperator{\Q}{\mathbb{Q}}
\DeclareMathOperator{\R}{\mathbb{R}}
\DeclareMathOperator{\K}{\mathbb{K}}
\DeclareMathOperator{\C}{\mathbb{C}}
\DeclareMathOperator{\N}{\mathbb{N}}
\DeclareMathOperator{\Z}{\mathbb{Z}}
\DeclareMathOperator{\grad}{grad}
\DeclareMathOperator{\I}{\mathcal{I}}

\DeclareMathOperator{\LT}{LT}

\DeclareMathOperator{\sos}{\mathtt{sos}}
\DeclareMathOperator{\hht}{ht}
\DeclareMathOperator{\quo}{quo}

\DeclareMathOperator{\lc}{lc}
\DeclareMathOperator{\sosone}{\mathtt{univsos1}}
\DeclareMathOperator{\sostwo}{\mathtt{univsos2}}
\DeclareMathOperator{\sosgrad}{\mathtt{sosgradient}}
\DeclareMathOperator{\sosshape}{\mathtt{sosgradientshape}}

\def\x{\bm{x}}
\def\y{\bm{y}}

\usepackage[normalem]{ulem}
\usepackage{todonotes}
\usepackage{xcolor}

\newcommand\msout{\bgroup\markoverwith
	{\textcolor{magenta}{\rule[.5ex]{2pt}{0.4pt}}}\ULon}

\begin{document}
	\maketitle
	\smallskip

	\begin{abstract}
    Assessing non-negativity of multivariate polynomials over the reals, through
    the computation of {\em certificates of non-negativity}, is a topical issue
    in polynomial optimization. This is usually tackled through the computation
    of {\em sums-of-squares decompositions} which rely on efficient numerical
    solvers for semi-definite programming.

    This method faces two difficulties. The first one is that the certificates
    obtained this way are {\em approximate} and then non-exact. The second one
    is due to the fact that not all non-negative polynomials are
    sums-of-squares.

    In this paper, we build on previous works by Parrilo, Nie, Demmel and Sturmfels
    who introduced certificates of non-negativity modulo {\em gradient ideals}.
    We prove that, actually, such certificates can be obtained {\em exactly},
    over the rationals if the polynomial under consideration has rational
    coefficients and we provide {\em exact} algorithms to compute them. We analyze
    the bit complexity of these algorithms and deduce bit size bounds of such
    certificates.
	\end{abstract}



	\section{Introduction}\label{sec:intro}

	We denote by $\Q$ (resp. $\R$) the field of rational (resp. real) numbers and
  by $\x$ the $n$-tuple of variables $(x_1,\dots,x_n)$. Let $\K$ be a field, we
  denote by $\K[\x]$ the polynomial ring with base field $\K$ and variables
  $\x$. For a polynomial $f$ of degree $d$ in $\Q[\x]$, we consider the
  problem of computing certificates of non-negativity of $f$ over $\R^n$. This
  is a central issue in polynomial optimization.

	\paragraph*{Prior works}
  Computing certificates of non-negativity is usually done by decomposing $f$ as
  a sum of squares (SOS) of polynomials or rational fractions. It is well-known
  that all non-negative univariate polynomials with real coefficients can be
  decomposed as a sum of squares of polynomials. Also, any non-negative
  univariate polynomial $f$ with rational coefficients can be decomposed as a
  {\em weighted} sum of squares with rational coefficients, i.e. $f = \sum_i
  c_i s_i^2$ where $s_i$ has rational coefficients and $c_i$ is a
  positive rational~\cite{landau1906,pourchet1971}. Further, by SOS
  decomposition with rational coefficients, we mean {\em weighted} SOS
  decompositions with rational coefficients. Several algorithms already compute
  such SOS decomposition with rational coefficients of non-negative
  univariate polynomials with rational coefficients (see \cite{schweighofer1999,
    chevillard2011}) and bit complexity and bit size estimates are given in
  \cite{univsos}.

	The multivariate case is more difficult. Following the seminal works by
  \cite{lasserre2001, parrilo2000}, hierarchies of semi\-definite programs yield
  \emph{approximations} of weighted SOS decompositions of positive polynomials.
  Several heuristics have been proposed to lift such approximations to exact SOS
  decompositions of the input polynomial, starting with \cite{peyrl2008} and
  followed by \cite{kaltofen2008,kaltofen2009,kaltofen2012}. Note that
  algorithms in \cite{kaltofen2008,kaltofen2012} allow us to compute SOS
  decompositions on some degenerate examples or compute SOS of rational
  fractions. 
	Complexity issues are studied through the prism of perturbation-compensation
	techniques to compute SOS decompositions in the interior of the SOS cone
	\cite{magron2018ACM,magron2018,magron2021}.

	Still, computing {\em exact} certificates of non-negativity is especially hard
  because of the two following reasons. The first one is that there exist
  non-negative polynomials which are not SOS, for example, Motzkin's polynomial
  and Robinson's polynomial. Moreover, Blekherman proved in
  \cite{blekherman2006} that there are many more non-negative polynomials in
  $\R[\x]$ than SOS polynomials. The second one is that, even if a given
  polynomial with rational coefficients is SOS, there is no guarantee that there
  exists an SOS decomposition involving rational coefficients, as established
  in~\cite{scheiderer2016}.

	Alternative certificates of non-negativity, for instance, SAGE/SONC
	polynomials
	\cite{magron2019,wang2020second} can also be used but they face similar
	issues
	to the ones met by SOS techniques when it comes with generality.

  Deciding non-negativity over an arbitrary semi-algebraic set of a polynomial
  $f\in \Q[\x]$ can be done exactly using computer algebra algorithms. The best
  complexities for such a decision procedure are achieved by algorithms making
  effective the so-called critical point method \cite{grigor1988,basu1998},
  further practical developments in \cite{bank2010,bank2001,bank2005,safey2003}
  and their applications in polynomial optimization in
  \cite{greuet2012,greuet2014,BGHS}. Note that, even if these algorithms are
  exact (i.e. their results are exact provided that no bug has been
  encountered), they do not provide a certificate assessing non-negativity which
  can be checked a posteriori since they are root-finding algorithms. Their
  complexities are exponential in the dimension of the ambient space as they
  reduce the input problem to computing finitely many critical points of some
  well-chosen maps, hence considering \emph{gradient ideals}.

  Hence, all in all, such gradient ideals can be used to reduce the dimension of
  the set over which certifying non-negativity can be done. Under some
  assumptions, this idea is translated in \cite{parrilo2002} to an algorithm
  assessing the non-negativity of a given $f\in \R[\x]$. Precisely, assuming the
  gradient ideal $\I_{\grad}(f)$ generated by all partial derivatives of $f$ is
  zero-dimensional and radical, and that $f$ reaches its minimum over $\R^n$,
  this algorithm computes an SOS decomposition of $f$ in the quotient algebra
  $\R[\x]/\I_{\grad}(f)$ (or, in other words, an SOS decomposition of $f$ modulo
  $\I_{\grad}(f)$), i.e., $f$ is written as
		\[
		c_1 s_1^2+\cdots+c_ks_k^2 +
		\sum_{i=1}^nq_i \frac{\partial f}{\partial x_i}
		\] where the $s_i$'s and the $q_i$'s lie in $\R[\x]$ and the $c_i$'s are
		positive in $\R$. A similar result slightly
		relaxing the above assumptions is given in \cite{nie2006}. Note that
		when $f$
		has coefficients in $\Q$, there is no given guarantee that an SOS
		decomposition
		of it in $\Q[\x]/\I_{\grad}(f)$ will have rational coefficients too
		(i.e., the
		$s_i$'s and the $q_i$'s have coefficients in $\Q$ and the $c_i$'s lie
		in $\Q$).


	\paragraph*{Contributions} We build on the results of \cite{parrilo2002,
			nie2006}, to investigate this issue when $f\in \Q[\x]$. We
			assume in
		the whole paper that the gradient ideal associated to $f$ is radical and
		zero-dimensional and that $f$ reaches its infimum over $\R^n$. We
		summarize
	our contributions as follows:
	\begin{itemize}
		\item Under the above assumptions, we prove (Theorem \ref{thm:main})
		that $f$
		is non-negative over $\R^n$ if and only if $f$ is an SOS of polynomials
with rational coefficients over the quotient ring
		$\Q[\x]/\I_{\grad}(f)$.
		Interestingly, Theorem \ref{thm:main} can be applied to
		Robinson's polynomial \cite{robinson73}, which is not an SOS of
		polynomials
		(see Example~\ref{ex:Rob}), as well as Scheiderer's polynomial
		\cite{scheiderer2016}, which is an SOS of polynomials with real
		coefficients
		but not an SOS of polynomials with rational coefficients (see
		Example~\ref{ex:Schei}).
  \end{itemize}
  The next problem we tackle is to design algorithms computing such certificates
  of non-negativity, estimate their bit complexity and the bit size of such
  certificates.

  {To measure the \textit{bitsize} of a polynomial with rational
    coefficients, we will use its height, defined as follows.} The bitsize of an
  integer $b$ is denoted by $\hht(b):=\lfloor\log_2(|b|)\rfloor+1$ with
  $\hht(0):= 1$, where $\log_2$ is the logarithm in base $2$. Given $a\in\Z$ and
  $b\in\Z$ with $b\neq 0$ and $\gcd(a,b)=1$, we define $\hht\left
    (\frac{a}{b}\right )=\hht(a)+\hht(b)$. For a non-zero polynomial $f$ with
  rational coefficients, the bitsize $\hht(f)$ is defined as the maximal bitsize
  of the non-zero coefficients of $f$.
	For two mappings
	$p,q:\N^m\to \R$, the expression
	``$p(v)=O(q(v))$" means that there exists $b\in \N$ such that $p(v)\leq b
	q(v)$, for all $v\in \N^m$. We use the notation $p(v)= \widetilde O (q(v))$
	in
	order to indicate that $p(v)=O(q(v)\log^kq(v))$ for some $k\in
	\N$.

  \begin{itemize}
  \item From the proof of Theorem \ref{thm:main}, we derive an algorithm
    (Algorithm \ref{alg:SOSdecom}), named $\sosshape$, to compute an SOS
    decomposition of polynomials modulo the gradient ideal of $f$. This
    algorithm can certify non-negativity of polynomials which cannot
    be tackled with a direct SOS approach. We also investigate the bit
    complexity of $\sosshape$. We prove that, given as input an $n$-variate
    polynomial $f\in \Q[\x]$ of degree $d$ with maximal bitsize of its
    coefficients $\tau$, $\sosshape$ uses $$ \footnotesize{\widetilde
      O}((\tau+n+d)^2d^{6n}+ (\tau+n+d)d^{6n+4})$$ boolean operations.
This is better than the complexity estimates given
			in \cite[Theorem 12]{magron2021}, where the reported number of
			boolean
			operations is:
			$ \footnotesize{\widetilde O}(\tau^2 (4 d + 2)^{15 n + 6})$.
		\item We design a variant of Algorithm $\sosshape$, named
			$\sosgrad$, and which, on input $f\in \Q[\x]$ as above,
			decomposes it as a sum of \emph{rational fractions} modulo the
			gradient
			ideal associated to $f$. We prove that this variant uses
			\[
			\footnotesize{\widetilde O}\left((\tau+n+d)d^{4n+4}\right).
			\]
			boolean operations and, consequently, has better complexity than
      Algorithm $\sosshape$.
	\end{itemize}
  Both algorithms have been implemented using the {\sc Maple} computer algebra
  system. We report on practical experiments showing that it can already assess
  the non-negativity of numerous polynomials which are out of reach of, e.g.,
  hybrid methods computing sums of squares decompositions such as
  \cite{magron2018ACM}.

	\paragraph*{Structure of the paper}
	In the next section, we
	recall basic notions and fundamental results used in the paper. In Section
	\ref{sec:pol}, we prove the existence of an SOS decomposition of polynomials
	modulo the gradient ideal of $f$, introduce Algorithm
	$\sosshape$ and analyze its bit complexity. The results for
	decomposing $f$ as an SOS of rational fractions modulo the gradient ideal
	are
	presented in Section \ref{sec:rat}. Practical experiments are given in the
	last section.

	\section{Preliminaries}\label{sec:pre}


	This section recalls basic notions and
	results from algebraic geometry, computational commutative algebra, and
	complexity analysis.
Further details can be found in \cite{cox2013}.

Let $\K$ be a field. An additive subgroup $\I$ of $\K[\x]$ is said
	to be an \textit{ideal} of $\K[\x]$ if $hg\in \I$ for any $h\in \I$ and
	$g\in\K[\x]$. 	Given $g_1,\dots,g_r $ in $ \K[\x]$, we denote by
	$\left\langle
	g_1,\dots,g_r \right\rangle$ the ideal generated by $g_1,\dots,g_r$. If $\I$
	is an ideal of $\K[\x]$ then, according to Hilbert's basis theorem
	(see,
	e.g.,
	\cite[Theorem~4]{cox2013}),
there exist $g_1,\dots,g_r \in \K[\x]$ such that
	$\I=\left\langle g_1,\dots,g_r \right\rangle$.

Let $\I$ be an ideal of $\R[\x]$. The
		algebraic variety associated to $\I$ is defined as
	$$V(\I):=\{x\in{\C}^n:\forall g\in \I, g(\x)=0\}.$$ The real algebraic
	variety
	associated to $\I$ is $V^{\R}(\I)=V(\I)\cap\R^n$. Recall that the ideal $\I$
	is \textit{zero-dimensional} if the
	cardinality of $V(\I)$ is finite, and that $\I$ is \textit{radical} if
	$$ g^k\in \I \text{ for some }k \in \N\Longrightarrow g \in \I.$$

	Let $f$ be a polynomial in $\R[\x]$.
	Recall that the \textit{gradient ideal} $\I_{\grad}(f)$ of $f$ is the ideal
	generated by all partial derivatives of $f$ in $\R[\x]$, i.e.,
	\begin{equation*}\label{f:partial}
		\I_{\grad}(f):=\left\langle\frac{\partial f}{\partial x_1} \,,\dots,
		\,\frac{\partial f}{\partial x_n} \right\rangle.
	\end{equation*}
	The (resp. \textit{real}) \textit{gradient variety}
associated to $f$ is respectively the (resp. real) algebraic variety
		associated to
		$\I_{\grad}(f)$. We denote them respectively by $V_{\grad}(f)$ and
	$V^{\R}_{\grad}(f)$. Let $\K$ be a real field contained in $\R$. One
	says that $f$ is a {(weighted)} \textit{sum of
		squares} (SOS) of polynomials in $\K[\x]$ if there exist polynomials
	$q_1,\dots,q_s$ in $\K[\x]$ and positive numbers $c_1,\dots,c_s$ in $\K$ such that
	$f=\sum_{j=1}^sc_jq_{j}^{2}$. Furthermore,
	$f$ is an SOS of
	polynomials over the quotient ring $\K[\x]/\I_{\grad}(f)$ if there exists
	$g\in \I_{\grad}(f)$ such that $f-g$ is SOS in $\K[\x]$, i.e., $f$ can be
	decomposed as follows:
	\begin{equation*}\label{f1:SOS}
		f=\sum_{j=1}^sc_jq_{j}^{2}+
		\sum_{i=1}^{n}\phi_i\frac{\partial f}{\partial x_i},
	\end{equation*}
	for some polynomials $q_1, \dots, q_s, \phi_1, \dots, \phi_s$ in $\K[\x]$ and positive numbers $c_1,\dots,c_s$ in $\K$ .

	Clearly, if $f$ is SOS over $\R[\x]/\I_{\grad}(f)$ then $f$ is non-negative
	over $V^{\R}_{\grad}(f)$. We recall below
	\cite[Theorem 1]{parrilo2002}.

	\emph{Let $f$ be a polynomial in $\R[\x]$. Suppose that the gradient ideal
		$\I_{\grad}(f)$ is zero-dimensional and radical. Then, $f$ is
		non-negative
		over
		$V^{\R}_{\grad}(f)$ if and only if $f$ is SOS over the quotient ring
		$\R[\x]/\I_{\grad}(f)$.}

	{We now recall useful results in the univariate case. It is well-known that
		$f\in \R[t]$ is non-negative over $\R$ if and only if $f$ is SOS. This
		property holds also for polynomials with coefficients in a subfield
		$\K$ of
		$\R$}. More precisely, we have the following theorem:

	\begin{theorem}[\cite{landau1906,pourchet1971}]\label{w_SOS}
	Let $\K$ be a subfield
		of $\R$ and $f\in \K[t]$. Then, $f$
		is non-negative over $\R$ if and only if $f$ admits a weighted SOS
		decomposition of polynomials in $\K[t]$, i.e., there exists a positive
		integer $s$,  non-negative numbers $c_1,\dots, c_s\in \K$ and
		polynomials
		$g_1,\dots,g_s\in \K[t]$, such that $f=\sum_{j=1}^sc_j g_{j}^{2}$.
	\end{theorem}

	Let $\K$ be a commutative field and $<$ be a monomial ordering on $\K[\x]$
	and $\I\neq \{0\}$ be an ideal. We denote
	by $\LT_<(\I)$ the set of all leading terms $\LT_<(g)$ of $g\in \I$, and  by
	$\left\langle \LT_<(\I)\right\rangle $ the ideal generated by the elements
	of
	$\LT_<(\I)$.

	A subset $G=\{g_1,\dots,g_r\}$ of $\I$ is said to be a \textit{Gr\"{o}bner
		basis} of $\I$ w.r.t. some monomial order $<$ if
    \[ \left\langle
	\LT_<(g_1),\dots,\LT_<(g_r)\right\rangle=\left\langle\LT_<(\I)\right\rangle.
	\]
	Note that every ideal in $\K[\x]$ has a Gr\"{o}bner basis. A Gr\"{o}bner
	basis
	$G$ is \textit{reduced}
	if
	the two following conditions hold: the leading coefficient of $g$ is $1$,
	for
	all $g\in G$; there are no monomials of $g$ lying in $\left\langle
	\LT_<(G)\setminus\{g\} \right\rangle$. Every ideal $\I$ has a unique reduced
	Gr\"{o}bner basis. We refer the reader to \cite{cox2013} for more details.
  Further, when the monomial order $<$ is clear from the context, we omit as a
  subscript in the above notation.

	Assume that $\I$ is a zero-dimensional and radical ideal in $\Q[\x]$ and
	that
	$G$ is the reduced Gr\"{o}bner basis of $\I$ with respect to the
  lexicographical order $x_1<_{lex} \cdots <_{lex} x_n$. One says that $\I$ is
   in \textit{shape position} if $G$ has the following shape:
	\begin{equation}\label{f:sysshape}
		G=[w,x_2-v_2,\dots,x_n-v_n],
	\end{equation}
	where $w,v_2,\dots,v_n$ are polynomials in
	$\K[x_1]$ and $\deg w=\sharp V(\I)$.  

	The following lemma, named Shape Lemma, gives us a criteria for the shape
	position of an ideal.

	\begin{lemma}[Shape Lemma, \cite{gianni89}]\label{lm:shape}
		Let $\I$ be a zero-dimensional and radical ideal and $<_{lex}$ be a
    lexicographic monomial order in $\Q[\x]$. If $V(\I)$ is the union of
    $\delta$ points in $\C^n$ with distinct $x_1$-coordinates, then $\I$ is in
    shape position as in \eqref{f:sysshape}, where $v_2,\dots,v_n$ are
    polynomials in $\Q[x_1]$ of degrees at most $\delta-1$.
	\end{lemma}


Let $V$ be a zero-dimensional algebraic subset of $\C^n$, $\delta:=\sharp
V$.
A \textit{zero-dimensional rational parametrization} $\mathcal{Q} =
((w,\kappa_1,\dots,\kappa_n),\lambda)$ of $V$ consists in $n+1$
univariate polynomials $w,\kappa_1,\dots,\kappa_n$ in $\Q[t]$,  where
$w'$ is the derivative of $w$, such that $w$ is monic and
squarefree,  $\deg \kappa_i < \deg w$, for $i=1,\dots,n$, and a
$\Q$-linear form $\lambda$ in $n$
variables satisfying $\lambda(\kappa_1,\dots,\kappa_n) = tw' \mod w$, such
that
$$V=\left\lbrace \Big(\frac{\kappa_1(t)}{w'(t)}, \, \dots, \,
\frac{\kappa_n(t)}{w'(t)}\Big) : w(t)=0\right\rbrace.$$
The condition on the linear form  $\lambda$ states that the roots of $w$
are precisely the values taken by $\lambda$ on $V$, and that $\lambda$
separates $V$, i.e.,  $\lambda(\x)\neq \lambda(\y)$ for any distinct pair
$\x,\y$ in $V$.

Let $f$ be in $\Q[\x]$ of degree $d$ and bitzise $\tau$. Assume that
$V_{\grad}(f)$ is finite. By applying \cite[Corollary 2]{safey2018} to the
system of partial derivatives, we obtain the following corollary (Corollary
\ref{cor:bit1}) which states that there exists an algorithm computing a
zero-dimensional rational parametrization of $V_{\grad}(f)$ and provides bit
complexity estimates for when applying the algorithm in \cite{safey2018} to
gradient ideals. The proof of Corollary \ref{cor:bit1} is straightforward
from \cite[Corollary 2]{safey2018} and is then postponed to
Appendix \ref{app:bitparam}.


\begin{corollary}\label{cor:bit1} Assume that
    $V_{\grad}(f)$ is finite. There exists an algorithm that
    takes $f$
    as in input, and
    that produces one of the following outputs:
    \begin{itemize}
        \item[\rm a)] either a zero-dimensional rational parametrization of
        $V_{\grad}(f)$;
        \item[\rm b)] or a zero-dimensional rational parametrization of
        degree less
        than that of $V_{\grad}(f)$;
        \item[\rm c)] or fails.
    \end{itemize}
    In any case, the algorithm uses
    \begin{equation}\label{f:Osim2}
        \widetilde O\Big(n^2(d+\tau)d^{2n+1}\binom{n+d}{d}\Big),
    \end{equation}
    boolean operations. Moreover, the polynomials
    $w,\kappa_1,\dots,\kappa_n$
    involved in the zero-dimensional rational parametrization output have
    degree at most $(d-1)^n$ and bitsize $ \footnotesize{\widetilde
        O}\left((d+\tau+n)(d-1)^{n}\right)$.
\end{corollary}

Assume that $\mathcal{Q} = ((w,\kappa_1,\dots,\kappa_n),x_1)$ is a
zero-dimensional rational parametrization of $V_{\grad}(f)$ given by
the
algorithm from Corollary \ref{cor:bit1}.  The following lemma
(Lemma \ref{lm:parame}) and its proof point out the explicit shape position of
 $\I_{\grad}(f)$. Moreover, the degree and the bit complexity of the involved
  polynomials are estimated.

\begin{lemma}\label{lm:parame}
There exist
       polynomials $w,v_2,\dots,v_n$ in $\Q[x_1]$ satisfying $\deg v_i<\deg w$,
       for $i=2,\dots,n$, such that
       $\I_{\grad}(f)=\left\langle
       w,x_2-v_2,
       \dots,x_n-v_n\right\rangle$.
{Furthermore, to compute
$w,v_2,\dots,v_n$, we
use
\begin{equation}\label{f:compv}
    \widetilde O\left( (\tau+n+d)^2d^{6n}\right)
\end{equation} boolean operations. Their degrees are at
most
$(d-1)^n$ and their maximal bitsizes are bounded from above by
$\footnotesize{\widetilde O}((\tau+n+d)d^{3n})$.}
\end{lemma}
\begin{proof} {Here we give only the proof of the degree estimate. The
         proof of
     the bit complexity is routine but rather technical and then postponed to
      Appendix \ref{app:q}.}

Because $w$ is squarefree and $w'$ is the derivative of $w$, one sees
	that the gcd of $w$ and $w'$ is
	$1$. From the extended Euclidean algorithm \cite[Algorithm
     3.14]{gathen2013},
     there exist two B\'{e}zout coefficients of $w,w'$,
     namely $a,b$ in $\Q[x_1]$, such that $aw+bw'=1$.
     For $i=2,\dots,n$, we see that $w'x_i(t)=\kappa_i(t)$ for any $t$
	satisfying $w(t)=0$. As $\deg \kappa_i \leq \deg w$ and the linear form
	$\lambda=x_1$ separates
	$V$, we
	have $w'x_i=\kappa_i$. This yields  $bw'x_i=b\kappa_i$. Since $bw'=1-aw$,
	we observe that $x_i-awx_i =b\kappa_i$
	and, hence, $x_i=b\kappa_i \mod w$. By denoting
	$v_i:=b\kappa_i \mod w$, we obtain  $w,v_2,\dots,v_n$ which
	are the desired polynomials.
\end{proof}


The two following lemmas establish the bit
	complexity of
	Euclidean division algorithm and the extended Euclidean algorithm for
	univariate
	polynomials over $\Z$ which will be used later on (in Proposition
	\ref{prop:bitshape}) to investigate the bit complexity of our algorithms.

	\begin{lemma}\label{lm:div1}
Let $a,b$ be polynomials in
		$\Z[t]$, with $\deg b=m \leq \deg a=d$, and $\tau$ an upper bound of
		$\hht(a)$ and $\hht(b)$. To compute the
		quotient
		$q$ and the remainder $r$ of the division of $a$ by $b$, we use the
		Euclidean division algorithm
		\cite[Algorithm 2.5]{gathen2013}.
		Then, this algorithm uses
			$\footnotesize{\widetilde
			O}\left(m\tau (d-m)^2\right)$ boolean operations. Furthermore, both
		bitsizes of $q$ and $r$ are bounded from above by
		$\footnotesize{\widetilde
			O}\left(\hht(d-m)\right)$.
	\end{lemma}

	Again, the proof of Lemma \ref{lm:div1} is routine but rather technical. We
  postpone it to Appendix \ref{app:lm:div1}.

{Denote by $\Q(x_1)$ the field of rational fractions in variable $x_1$ with
  coefficients in $\Q$. With the lexicographic monomial order
    $x_2<\dots<x_n$, we consider the standard (multivariate) division \cite[Ch.
    2, Sec 3.]{cox2013} of $g \in \Q[x_1][x_2,\dots,x_n]$ by the list
    $[x_2-\frac{a_2}{a_0},\dots,x_n-\frac{a_n}{a_0}]$, with
    $a_0,a_2,\dots,a_n \in \Q[x_1]$.
    To compute the quotients
    $\phi_2,\dots,\phi_n\in \Q(x_1)[x_2,\dots,x_n]$ and remainder $r\in \Q(x_1)$ such that
    $g=\sum_{i=2}^n\phi_{i}(x_i-\frac{a_i}{a_0})+r$, we iterate
    classical univariate divisions by $x_i-\frac{a_i}{a_0}$ for $2\leq i \leq n$
    considering them as univariate in $x_i$ so that we eliminate step by step
    the variables $x_2, \ldots, x_n$ in $g$.
    The details of this algorithm,
    which we name \textsf{Eliminate}, are given in Appendix~\ref{app:divshape}
    (Algorithm \ref{alg:divshap}).
The inputs of $\mathtt{Eliminate}$ are $g,a_0,a_2,\dots,a_n$ and the output is the list $[\phi_2,\dots,\phi_n]$ and the remainder $r$.
\if Note that when we perform $\mathtt{Eliminate}(g,a_0,a_2,\dots,a_n)$ and $\mathtt{Eliminate}(g+h,a_0,a_2,\dots,a_n)$, with $h\in \Q[x_1]$, we obtain the same quotients in the outputs. \fi}

The bit complexity of \textsf{Eliminate} is given in the following
lemma whose proof (which is quite routine) is given in Appendix
\ref{app:divshape}.

\begin{lemma}\label{lm:divshape}
Assume that $g\in \Q[x_1][x_2,\dots,x_n]$ has degree $d$ in $x_2,\dots,x_n$ and bitsize $\tau_g$, and that the
		polynomials $a_0,a_2,\dots,a_n \in \Q[x_1]$ have
		bitsizes at most $\tau_{a}$.
		Then, Algorithm $\mathtt{Eliminate}$ runs in
		$$\footnotesize{\widetilde
			O}\left(n\tau_{g}+n^2d\tau_a\right)$$
		boolean operations and the bitsizes of the outputs $\phi_2,\dots,\phi_n$
are in $\footnotesize{\widetilde O}\left(\tau_g+nd\tau_a\right)$.
\end{lemma}


\section{SOS of polynomials modulo gradient ideals}\label{sec:pol}

	\subsection{The existence of an SOS decomposition over the rationals}
The main result of this section is stated below.

	\begin{theorem}\label{thm:main}
		Let $f\in \Q[\x]$ such that the following conditions hold:
		\begin{itemize}
			\item[\rm a)] The infimum $f^{\star}=\inf \{f(x):x\in \R^n\}$ is
			attained.
			\item[\rm b)] The gradient ideal $\I_{\grad}(f)$ is
			zero-dimensional and radical.
		\end{itemize}
		Then, $f$ is non-negative over $\R^{n}$ if and only if $f$ is an SOS of
    polynomials over the quotient ring $\Q[\x]/\I_{\grad}(f)$.
	\end{theorem}

	\begin{proof} Suppose that $f$ is non-negative over $\R^{n}$ and $\sharp
    V_{\grad}(f)=\delta$. We prove that $f$ is an SOS of
    polynomials over the quotient ring $\Q[\x]/\I_{\grad}(f)$. We
    consider the two following cases:

{\sc Case 1.} Distinct points in $V_{\grad}(f)$ have distinct
	$x_1$-coordinates.
Consider the lexicographic monomial order $x_1 <x_2<\dots<x_n$ on
		$\Q[\x]$.
		Since the gradient ideal is zero-dimensional and radical, according to
		Shape Lemma
		(Lemma~\ref{lm:shape}), the reduced Gr\"{o}bner basis of
		$\I_{\grad}(f)$ has
		the following shape:
		\begin{equation}\label{f:shape}
			[w,x_2-v_2,\dots,x_n-v_n ],
		\end{equation}
		where $v_2,\dots,v_n$ are polynomials in $\Q[x_1]$ of degree at most
		${\delta}-1$. We denote
		\begin{equation}\label{f:h}
			h(x_1):=f(x_1,v_2,\dots,v_n),
		\end{equation}
		where $x_i$ is replaced by $v_i$ in $f$ for $i=2,\dots,n$.
With the order $<$, we divide $f-h$ by the system in \eqref{f:shape} by using
the division algorithm in \cite[Ch. 2, Sec 3.]{cox2013}.
Then, there
    exist $\phi_1,\dots,\phi_n$ in $ \Q[\x]$, and $r$ in $\Q[x_1]$ such that
		\begin{equation}\label{e2}
			f-h=\phi_1w+
			\sum_{i=2}^{n}\phi_i(x_i-v_i)+r,
		\end{equation}
		with $\deg r<{\delta}$. Let $x$ be in $V_{\grad}(f)$. From
		\eqref{f:h}
			and \eqref{e2}, one sees that $f(x)=h(x)$. Hence, $f-h$
			vanishes on $V_{\grad}(f)$.
Clearly, the value of $\phi_1w+
			\sum_{i=2}^{n}\phi_i(x_i-v_i)$ is zero on
			$V_{\grad}(f)$.
			This implies that $r$ also
			vanishes on the image set
    $\pi(V_{\grad}(f))$, where $\pi(x_1,\dots,x_n)=x_1$. Since distinct points
    in $V_{\grad}(f)$ have distinct $x_1$-coordinates, it
      holds that $\sharp\pi(V_{\grad}(f))=\sharp V_{\grad}(f) ={\delta}$. As
    $\deg r<{\delta}$, we conclude that $r\equiv 0$. Hence, from $\eqref{e2}$,
    we obtain the following representation:
		\begin{equation}\label{e3}
			f=h+\phi_1w+
			\sum_{i=2}^{n}\phi_i(x_i-v_i).
		\end{equation}

The set $\{(x_1,x_2,\dots,x_n)\in\R^n:x_2 = v_2, \ldots, x_n=v_n\}$ defines
    a curve which is parametrized by $x_1$. Recall that
    $f$ is non-negative over $\R^n$. Hence $f$ is non-negative over this curve.
    Since $f$ takes the same values over this curve as $h$ takes over $x_1$ when
    $x_1$ ranges in $\R$, one can conclude that the univariate polynomial $h$ is
    also non-negative over $\R$. According to the results on SOS decompositions
    of univariate polynomials with rational coefficients
    in Theorem \ref{w_SOS}
$h$ is a sum of $s$ squares in $\Q[x_1]$, i.e., there exist $q_1,\dots,q_s
\in \Q[x_1]$ and $c_1,\dots,c_s$ in $\Q_+$
    such that $h=c_1q^2_1+\dots+c_sq_s^2$. Therefore, from $\eqref{e3}$, we assert
    that $f$ is an SOS of polynomials over $\Q[\x]/\I_{\grad}(f)$.

{\sc Case 2.} There are two distinct points in $V_{\grad}(f)$ such that
their $x_1$'s-coordinates are equal. 
According to \cite[Lemma 2.1]{rouillier1999},
there is
	$j\in\{1,\dots,(n-1)\delta(\delta-1)/2\}$ such that the linear function
	$u:=x_1+jx_2+\dots+j^{n-1}x_n$ separates $V_{\grad}(f)$, i.e., $u(x)\neq
	u(y)$ for any distinct points $x,y$ in $V_{\grad}(f)$. We consider the
	change of variables $\y=T\x$, where
	\begin{equation}\label{T}
		T=\begin{bmatrix}
			1 & j & j^2 &\cdots & j^{n-1} \\
			0 & 1 & 0 & \cdots  & 0 \\
			0& 0 & 1 & \cdots & 0\\
			\vdots & \vdots & \vdots & \ddots & \vdots\\
			0& 0 & 0 & \cdots & 1
		\end{bmatrix}.
	\end{equation}
We see that $T$ is an invertible matrix. Then we obtain a polynomial
$g(\y)=f(T^{-1}\y)$ in variables $y_1, y_2, \dots, y_n$ having the following
property: the infimum $g^{\star}=\inf \{g(y):y\in \R^n\}$
is attained. Because of the chain rule $\nabla g= \nabla f \circ
T^{-1}$,
we have $$V_{\grad}(g)=\{y\in\C^n:y=Tx, x\in V_{\grad}(f)\};$$
Thus, the gradient ideal $\I_{\grad}(g)$ is zero-dimensional
		and radical. Moreover,
		since
		$y_1=u(\x)$ separates
		$V_{\grad}(f)$,
		distinct points in $V_{\grad}(g)$ have distinct
		$y_1$-coordinates. 
We observe that $g\in \Q[\y]$ is non-negative and satisfies the conditions of
the theorem; Case 1 happens to $V_{\grad}(g)$ as well. Hence, there exists
an SOS decomposition of $g$ modulo $\I_{\grad}(g)$
\begin{equation}\label{f1:gSOS}
	g(\y)=\sum_{j=1}^sc_j\bar q_{j}^{2}(\y)+
	\sum_{i=1}^{n}\bar\phi_i(\y)\frac{\partial
	g}{\partial y_i},
\end{equation}
where $\bar\phi=(\bar\phi_1,\dots,\bar\phi_n)$ and $\bar\phi_i\in \Q[\y]$.
In \eqref{f1:gSOS}, we replace $\y$ by $T\x$ and $\frac{\partial
	g}{\partial y_i}$ by $\frac{\partial
	f}{\partial x_i} \circ
T^{-1}$, we obtain an decomposition of $f$ as
follows:
\begin{equation}\label{f1:fSOS}
	f(\x)=g(T\x)=\sum_{j=1}^sc_j\bar q_{j}^{2}(T\x)+
	\sum_{i=1}^{n}\bar\phi_i(T\x)\frac{\partial
		f}{\partial x_i}\circ T^{-1}.
\end{equation}
Because of $(\frac{\partial
	f}{\partial x_i}\circ T^{-1})(Tx)=\frac{\partial
	f}{\partial x_i}(x)$, \eqref{f1:fSOS} is an SOS decomposition of $f$ modulo
 $\I_{\grad}(f)$ of $f$.


We now prove the reverse conclusion. Suppose that $f$ is SOS over the
		quotient
		ring $\Q[\x]/\I_{\grad}(f)$, i.e., $f$ can be decomposed as follows:
		\begin{equation}\label{f:SOSstar}
			f=\sum_{j=1}^sc_jq_{j}^{2}+
			\sum_{i=1}^{n}\phi_i\frac{\partial f}{\partial x_i},
		\end{equation}
		for some polynomials $q_1, \dots, q_s, \phi_1, \dots, \phi_n \in
		\Q[\x]$, and $c_1,\dots,c_s$ in $\Q_+$.
Let $x^{\star} \in \R^{n}$ be such that $f(x^{\star})=f^{\star}$.
		Then $x^{\star}$ is a critical point of $f$ over $\R^{n}$, i.e.,
		$x^{\star}$
		belongs to the variety $V_{\grad}(f)$; thus, we have
		$$\sum_{i=1}^{n}\phi_i(x^{\star})\frac{\partial f}{\partial
			x_i}(x^{\star})=0.$$ From \eqref{f:SOSstar}, we see that
    $f(x^{\star})=\sum_{j=1}^sc_jq_{j}^{2}(x^{\star})$ and so this value is
    non-negative. By assumption, for all $x$ in $\R^{n}$,
    $f(x)\geq f(x^{\star})$. Hence, $f$ is non-negative over $\R^n$.
	\end{proof}

	\begin{remark} Assume that $\rm Q$ is a real field and $\rm R$ is the real
    closure of $\rm Q$. All arguments in the proof of Theorem \ref{thm:main} can
    be applied for $f$ in ${\rm Q}[\x]$. Hence, the conclusion of Theorem
    \ref{thm:main} holds for the case ${\rm Q}[\x]$, i.e., $f$ is non-negative
    over ${\rm R}^n$ if and only if $f$ is an SOS of polynomials over the
    quotient ring ${\rm Q}[\x]/\I_{\grad}(f)$ provided that the two following
    conditions hold: the infimum $f^{\star}=\inf \{f(x):x\in {\rm R}^n\}$
    is attained;
    the gradient ideal
    $\I_{\grad}(f)$ is zero-dimensional and radical.
	\end{remark}

	\begin{remark}\label{rmk:phi1} In the proof of Theorem \ref{thm:main}, one
	can	see that  $f-h$ vanishes not only on
		$V_{\grad}(f)$ but also on the variety defined by $\left\langle
		x_2-v_2,\dots,x_n-v_n\right\rangle $.
		Hence, $\phi_1$ in \eqref{e3}  is zero and
		\eqref{e3}
becomes $f=c_1q^2_1+\dots+c_sq_s^2+\sum_{i=2}^{n}\phi_i(x_i-v_i)$.
	\end{remark}

	\begin{remark} The condition $\rm a)$ in Theorem \ref{thm:main} is used
	only to
		prove the reverse conclusion. Therefore, even without this condition,
		the
		following assertion still holds: if $\I_{\grad}(f)$ is zero-dimensional
		radical
		and  $f$ is non-negative over $\R^{n}$, then $f$ is SOS modulo
		$\I_{\grad}(f)$.
	\end{remark}


Theorem \ref{thm:main} provides certificates of non-negativity
for polynomials in $\Q[\x]$ which satisfy its
  assumptions and which are not SOS of polynomials with
  real (or rational) coefficients. We illustrate this with two examples.

  	\begin{example}\label{ex:Rob} We recall a polynomial of  Robinson
  	\cite{robinson73} that is non-negative but cannot be represented
  	as an SOS of polynomials,
  	$$\bar f_R=x_1^6+x_2^6+x_3^6-x_1^4x_2^2-x_1^4x_3^2-x_2^4x_1^2-
  	x_2^4x_3^2-x_3^4x_1^2-x_3^4x_2^2+3x_1^2x_2^2x_3^2.$$
  	By substituting the third variable $x_3$ by $1$ in $\bar f_R$,
  	we get the following non-negative polynomial:
  	$$f_R=x_1^6+x_2^6-x_1^4x_2^2+3x_1^2x_2^2-x_1^2x_2^4-
  	x_1^4-x_2^4-x_1^2-x_2^2+1.$$
Because $\bar f_R$ is the homogenization of $f_R$, $f_R$ cannot be
 represented as an SOS of polynomials \cite[Proposition 1.2.4]{marshall2008}.
The gradient ideal $\I_{\grad}(f_R)$ is zero-dimensional and radical.
Therefore, Theorem \ref{thm:main} tells us that $f_R$ is an SOS of polynomials
  	over the quotient ring $\Q[\x]/\I_{\grad}(f_R)$.
  \end{example}

	\begin{example}\label{ex:Schei} In \cite{scheiderer2016}, Scheiderer
	introduced
		the following homogeneous polynomial:
		$$\bar f_S=x_1^4+x_1x_2^3+x_2^4-3x_1^2x_2x_3-4x_1x_2^2x_3+
		2x_1^2x_3^2+x_1x_3^3+x_2x_3^3+x_3^4,$$
		that can be decomposed as an SOS of polynomials with algebraic
		coefficients but
		cannot be decomposed as an SOS of polynomials with rational
		coefficients. By
		replacing the third variable $x_3$ by $-1$, we obtain the non-negative
		polynomial
		$$f_S=x_1^4+x_1x_2^3+x_2^4+3x_1^2x_2+4x_1x_2^2+2x_1^2-x_1-x_2+1.$$
Note that the conclusion in \cite[Proposition 1.2.4]{marshall2008}
	holds for polynomials with rational coefficients, i.e., $g\in\Q[\x]$ is SOS
	in $\Q[\x]$ if and only if its homogenization so is in $\Q[\x]$. Hence,
	the polynomial $f_S$ is also SOS with algebraic coefficients but not SOS
	with rational ones. The gradient ideal $\I_{\grad}(f_S)$ satisfies
    the zero-dimensional and radical condition. Hence, according to Theorem
    \ref{thm:main}, $f_S$ is an SOS of polynomials over the quotient ring
    $\Q[\x]/\I_{\grad}(f_S)$.
	\end{example}

	An explicit SOS decomposition of $f_S$ will be given in
the next section.

	\subsection{Description of the algorithm}
	Based on the proof of Theorem \ref{thm:main}, we design an
    algorithm to
 compute an SOS decomposition of polynomials modulo the gradient ideal of
  a non-negative polynomial with rational coefficients.

	The input of $\sosshape$ is a non-negative polynomial
	$f\in\Q[\x]$ whose gradient ideal $\I_{\grad}(f)$ is zero-dimensional
and radical
	and satisfies the Shape Lemma assumption, i.e., all points in
	$V_{\grad}(f)$
	have distinct $x_1$-coordinates.

  The output includes the cardinality $\delta$ of $V_{\grad}(f)$, the lists of
  polynomials and numbers
  \[
   [w,v_2,\dots,v_n], \ [q_1,\dots,q_s],
    [\phi_2,\dots,\phi_n] \subset \Q[\x] , \text{ and } [c_1,\dots,c_s]
     \subset \Q_+
  \]
satisfying the relation
  \[
    f=\sum_{j=1}^sc_jq_{j}^{2}+ \sum_{i=2}^{n}\phi_i(x_i-v_i)
  \]
  In Step~1, we compute the reduced Gr\"{o}bner basis $G$ for $\I_{\grad}(f)$ by
  relying on a zero-dimensional rational parametrization of $V_{\grad}(f)$
  mentioned in Lemma~\ref{lm:parame}. In Step 2, we
    compute the quotients
  $\phi_2,\dots,\phi_n$ and the remainder $r$ of the division of $f$ by $G$. In Step 3, we
  compute a rational weighted SOS decomposition of
  the non-negative univariate polynomial $h$ by using Algorithm
  $\sosone$ or Algorithm $\sostwo$ described in \cite[Fig.
  1]{univsos} or \cite[Fig. 2]{univsos}, respectively.

	\begin{algorithm}
		\caption{Computing SOS of polynomials modulo the gradient
			ideal}\label{alg:SOSdecom}

		$\sosshape:=\text{proc}(f)$

		\textbf{Input:}  $f\in \Q[\x]$ non-negative over $\R^n$ such that
		$\I_{\grad}(f)$ is zero-dimensional and radical and all points in
		$V_{\grad}(f)$ have distinct $x_1$-coordinates

		\textbf{Output:} $\delta$ in $\N$, $[q_1,\dots,q_s]$,
		$[w,v_2,\dots,v_n] \subset \Q[x_1]$,
		$[\phi_2,\dots,\phi_n] \subset \Q[\x]$, and $[c_1,\dots,c_s] \subset \Q_+$ satisfying
		\begin{equation}\label{f:SOSQ}
			f=\sum_{j=1}^sc_jq_{j}^{2}
			+\sum_{i=2}^{n}\phi_i(x_i-v_i).
		\end{equation}

		\begin{itemize}
    \item [\rm 1:] Compute the reduced Gr\"{o}bner
      basis $G=[w,x_2-v_2,\dots,x_n-v_n]$ of $\I_{\grad}(f)$, with the
      lexicographical ordering $x_1< x_2< \cdots < x_n$
			\item [\rm 2:] {Compute the quotients $[\phi_2,\dots,\phi_n]$ and remainder $h$ of the division of $f$ by $G$ by performing
             $\mathtt{Eliminate}(f,1,v_2,\dots,v_n)$}
			\item [\rm 3:] Compute a rational weighted SOS decomposition
			$h=c_1q^2_1+\dots+c_sq_s^2$
			\item [\rm 4:] Return $\delta$, $[q_1,\dots, q_s]$,
			$[\phi_2,\dots,\phi_n]$, $[w,v_2,\dots,v_n]$, and $[c_1,\dots,c_s]$
		\end{itemize}
	\end{algorithm}

	\begin{remark}
		Suppose that the Shape Lemma assumption does not hold for
         $\I_{\grad}(f)$,
    i.e., there are two distinct points in $V_{\grad}(f)$ such that their
    $x_1$'s-coordinates are equal. As mentioned in the proof of Theorem
    \ref{thm:main}, we can find an invertible matrix $T$ given by \eqref{T},
    change of variables $\y=T\x$, and assign $g(\y):=f(T^{-1}\y)$. Here, we have
    $y_1=x_1+jx_2+\dots+j^{n-1}x_n$ for some $j>0$ and $y_i=x_i$ for
    $i=2,...,n$. We get a new non-negative polynomial in $n$ new variables with
    rational coefficients $g(\y)$ whose gradient ideal satisfies Shape Lemma
    assumption. Now we can apply Algorithm~$\sosshape$ for $g(\y)$ and obtain
    the output: the number $\bar \delta$, two lists $[\bar q_1,\dots,\bar q_s]$,
    $[\bar w,\bar v_2,\dots,\bar v_n]$ of polynomials in $\Q[y_1]$, a list
    $[\bar \phi_1,\dots,\bar \phi_n]$ of polynomials in $\Q[\y]$, and a list $[c_1,\dots,c_s] \subset \Q_+$. Since $\sharp
    V_{\grad}(f)=\sharp V_{\grad}(g)$, one has $\bar \delta=\delta$. The new
    polynomial $g$ can be decomposed as follows:
		$$g(\y)=\sum_{j=1}^sc_j\bar
		q_{j}^{2}(y_1)+\bar\phi_1(\y)\bar
		w(y_1)+\sum_{i=2}^{n}\bar\phi_i(\y)(y_i-\bar
		 v_i(y_1)).$$
		Hence, $f$ can be decomposed as:
		\begin{equation}\label{f:fx}
			f(\x)=\sum_{j=1}^sc_j\bar
			q_{j}^{2}(u(\x))+\bar\phi_1(T\x)\bar w(u(\x)) \medskip \\
			+\sum_{i=2}^{n}\bar\phi_i(T\x)(x_i-\bar v_i(u(\x))),
		\end{equation}
		where $u(\x)=x_1+jx_2+\dots+j^{n-1}x_n$.
		Clearly, $[w(u),x_2-\bar v_2(u),\dots,x_n-
		\bar v_n(u)]$ is also a basis for $V_{\grad}(f)$. Hence, \eqref{f:fx}
		provides us an SOS decomposition of $f$ modulo the gradient ideal of
		$f$.
	\end{remark}

	\begin{theorem}\label{thm:corectness} Let $f$ be a non-negative polynomial
	in
		$\Q[\x]$. Suppose that $f$ is  non-negative over $\R^n$,
		$\I_{\grad}(f)$ is
		zero-dimensional and radical, and all points in $V_{\grad}(f)$ have
		distinct
		$x_1$-coordinates. On input $f$, Algorithm $\sosshape$
		terminates and computes an SOS decomposition of $f$ modulo
      $\I_{\grad}(f)$ with rational coefficients.
	\end{theorem}
	\begin{proof} Assume that $f\in\Q[\x]$ is non-negative over $\R^n$ and its
		gradient ideal is zero-dimensional and radical. Here, we use the
		lexicographic
		monomial order $x_1<x_2<\dots<x_n$. Because the  Shape Lemma's
		assumption
		holds, the reduced Gr\"{o}bner basis of $\I_{\grad}(f)$  in Step~1  has
		the
		form $G=[w,x_2-v_2,\dots,x_n-v_n]$, and can be computed by
		using a zero-dimensional rational parametrization of $V_{\grad}(f)$ as
		in
		Lemma~\ref{lm:parame}. {In Step 2, we
		compute
		the quotients $[\phi_2,\dots,\phi_n]$ and the remainder $r$ of the division of $f$ by $G$
		by performing $\mathtt{Eliminate}(f,1,v_2,\dots,v_n)$
        (as in Algorithm~\ref{alg:divshap}). Here, we see that $r$ coincides with $h$, where $h=f(x_1,v_2,\dots,x_n)$ as in the proof of Theorem \ref{thm:main}, because of
        $$r=f-\sum_{i=2}^{n}\phi_i(x_i-v_i)=h.$$
}
        \if Note that $h$ depends only on $x_1$, so $[\phi_2,\dots,\phi_n]$ can
         be obtained by performing $\mathtt{Eliminate}(f,1,v_2,\dots,v_n)$. \fi
        In Step 3, the univariate polynomial
		$h$ is non-negative with rational coefficients, so by using
		$\sosone$ or
		$\sostwo$ \cite{univsos}, we can compute an SOS decomposition
		of
		$h=c_1q^2_1+\dots+c_sq_s^2$. Hence, according to the proof of
		Theorem~\ref{thm:main}, we get \eqref{f:SOSQ} which is an SOS
		decomposition
		modulo the gradient ideal of $f$.
	\end{proof}

	To illustrate how the algorithm works, we consider the following simple
	example.

	\begin{example}
		Consider the polynomial $f(x_1,x_2)= 2x_1^4+2x_1x_2+x_2^2+10$. This
		polynomial
		is non-negative over $\R^n$. Firstly, the gradient ideal
		$\I_{\grad}(f)$
		is given
		by $\I_{\grad}(f)=\left\langle 8x_1^3 + 2x_2, 2x_1 + 2x_2
		\right\rangle$ which
		is zero-dimensional and radical with $\delta=3$. We compute the
		reduced
		Gr\"{o}bner basis of $\I_{\grad}(f)$, namely
		$\left\langle x_1^3-\frac{1}{4}x_1, x_2+x_1\right\rangle$, here
		$v_2(x_1)=-x_1$. Secondly, with the order $x_1<x_2$,
		the
		quotients of the division of $f$ by the Gr\"{o}bner basis are
		$\phi_1=0$
		and $\phi_2=x_1+x_2$, {and the remainder is given by
        $h(x_1)=f(x_1,v_2)=2x_1^4-x_1^2+10$.} Thirdly, by using Algorithm
		$\sostwo$  in \cite{univsos}, one gets an SOS decomposition of
$h=\frac{1}{2}x_1^4+\frac{3}{2}(x_1^2-\frac{5}{2})^2+\frac{13}{2}x_1^2+
\frac{5}{8}$.
		Finally, we obtain the following SOS decomposition of $f$ modulo its
		gradient ideal:
		\[f=\frac{1}{2}x_1^4+\frac{3}{2}\Big(x_1^2-\frac{5}{2}\Big)^2+
		\frac{13}{2}x_1^2+\frac{5}{8}+
		(x_1+x_2)\times \left({x_2+x_1}\right). \]
	\end{example}

	\subsection{Bit complexity analysis}
	This subsection investigates the bit complexity of $\sosshape$. Assume that
  $d$ and $\tau$ are respectively the degree and an upper bound of the bitsize
  of the coefficients of $f\in \Q[\x]$. We provide estimates for the bitsizes of
  polynomials in the output of $\sosshape(f)$ as well as for the number of
  required boolean operations required to execute it.

	We use Algorithm $\sosone$ in \cite[Fig. 1]{univsos} or Algorithm
	$\sostwo$ in \cite[Fig. 2]{univsos} to compute an SOS decomposition
	of
	the non-negative univariate polynomial $h$. The corresponding bit
	complexities are given as follows:

\begin{proposition}\label{prop:bitsos1} Let $v_2,\dots,v_n$ be as in
Lemma~\ref{lm:parame} and  $h(x_1)=f(x_1, v_2,\dots,
		v_n)$. To compute an SOS decomposition of $h$, Algorithm
		$\sosone$ and Algorithm $\sostwo$  run in
		\begin{equation}\label{f:bitsos1}
			\widetilde O\Big((
			d^{n+1}/2)^{3d^{n+1}/2}(\tau+n+d)d^{3n+1}\Big)
		\end{equation}
		and
		\begin{equation}\label{f:bitsos2}
		\widetilde O\left((\tau+n+d)d^{6n+4}\right)
		\end{equation}
		boolean operations, respectively.
	\end{proposition}

	\begin{proof}  Let $\tau_v = \max_{i}\{\hht(v_i)\}$. Lemma
		\ref{lm:parame} states that the bitsize of $\tau_v$ is bounded from
		above by $\footnotesize{\widetilde O}\left((\tau+n+d)d^{3n}\right)$,
		and that the polynomials $w,v_2,\dots,v_n$ have degree at most
		$(d-1)^n$.
		Since $\deg f=d$ and $h(x_1)=f(x_1,v_2,\dots,v_n)$, the
		degree of $h$ is at most $d(d-1)^{n}$.

		Let $\beta$ be the minimal common denominator of all non-zero
         coefficients
    of $h$. Computing an SOS decomposition of $h$ boils down to computing an SOS
    decomposition of $\beta h$. In particular, the execution time of $\sosone$
    (resp., $\sostwo$) on $h$ is the same as for $\beta h$.
		Now we estimate the bitsize of the polynomial $\beta h\in\Z[x_1]$. By
		the definition of $h$, we observe that $\hht(h)\leq \tau+d\tau_v.$
		It follows that
		$\hht(\beta h)\leq \hht(\beta)+\tau+d\tau_v$. By definition we have
		$\hht(\beta)\leq \tau+d\tau_v$.
		This yields
		\begin{equation}\label{f:tauah}
			\hht(\beta h)\leq 2\left(\tau+d\tau_v\right).
		\end{equation}
		From \eqref{f:tauah} and above results, we obtain the following bitsize
		estimate for $\beta h$: $$\footnotesize{\widetilde
			O}\left(2(\tau+d(\tau+n+d)d^{3n})\right)=\footnotesize{\widetilde
			O}\left((\tau+n+d)d^{3n+1}\right).$$

		To compute an SOS decomposition of $\beta h$, we rely on
		$\sosone$ or $\sostwo$. From \cite[Theorem~17]{univsos} and
		\cite[Theorem 24]{univsos}, the boolean
		running time
		of 	$\sosone$ corresponds to the quantity given by
		\eqref{f:bitsos1}. If we use $\sostwo$, then the number
		of boolean operations will be bounded from above by
		$$\footnotesize{\widetilde
			O}\left(d^4(d-1)^{4n}+d^4(\tau+n+d)(d-1)^{6n}\right),$$
		which can be further reduced to \eqref{f:bitsos2}.
	\end{proof}

	\begin{proposition}\label{prop:bitshape} Let $v_2,\dots,v_n$ be as
		in Proposition \ref{prop:bitsos1}. To compute the list
		$\phi_2,\dots,\phi_n$ in the output of Algorithm
		$\sosshape$, Algorithm $\mathtt{Eliminate}$
        runs in
		$\footnotesize{\widetilde O}(n^2(\tau+n+d)d^{3n+1})$ boolean operations
		and the bitsizes of $\phi_2,\dots,\phi_n$ are $\widetilde
		O\left(n(\tau+n+d)d^{3n+1}\right)$.
	\end{proposition}
	\begin{proof}  From Lemma \ref{lm:parame}, the polynomial $v_i$
		has 
		bitsize at most
		$\footnotesize{\widetilde O}((\tau+n+d)d^{3n})$.
\if Hence, the degree of $h$ is at most $d(d-1)^{n}$ and
the bitzise of $h$	is at most $\footnotesize{\widetilde
			O}\left(d(\tau+n+d)d^{3n}\right)$. \fi
	{We divide $f$ by $[x_2-v_2,\dots,x_n-v_n]$ while  performing
		$\mathtt{Eliminate}(f,1,v_2,\dots,v_n)$ as in
		Algorithm~\ref{alg:divshap} to obtain the list of quotients
        $[\phi_2,\dots,\phi_n]$ and the remainder $h=h(x_1,v_2,\dots,v_n)$.}
Applying Lemma \ref{lm:divshape} for this division, we conclude that
Algorithm $\mathtt{Eliminate}$ runs in
		$\footnotesize{\widetilde O}(n^2(\tau+n+d)d^{3n+1})$ boolean
		operations, the estimate for the bitsize of $\phi_i$ is $\widetilde
		O\left(n(\tau+n+d)d^{3n+1}\right)$ as claimed.
	\end{proof}

	We are now ready to analyze the bit complexity of Algorithm
	\ref{alg:SOSdecom}.

	\begin{theorem}\label{thm:bit}
		Let $f \in \Q[\x]$ of degree $d$ and let $\tau$ be the maximum bitsize
         of its coefficients. Assume that the two conditions in Theorem
          \ref{thm:main} hold. Then, on input $f$, Algorithm $\sosshape$ runs in
		\begin{equation}\label{f:bitAlg41a}
			\widetilde O\Big((\tau+n+d)^2d^{6n}
			+(\tau+n+d)d^{3n+1}(d^{n+1}/2)^{3d^{n+1}/2}\Big)
		\end{equation}
		or
		\begin{equation}\label{f:bitAlg41b}
		\widetilde O\Big((\tau+n+d)^2d^{6n}+
			(\tau+n+d)d^{6n+4}\Big)
		\end{equation}
		boolean operations if in Step 3 we use Algorithm $\sosone$ or
		Algorithm $\sostwo$, respectively.
	\end{theorem}
	\begin{proof} Assume that in Step 3 we use $\sosone$ to compute an
	SOS
		decomposition of  $h$. Then, the number of boolean operations that
		$\sosshape$ uses to compute the SOS decomposition of $f$ is
		the sum of the four following ones:
    \begin{enumerate}
    \item the number of boolean operations required to compute the
      zero-dimensional rational parametrization $\mathcal{Q}$ of $V_{\grad}(f)$
      as in \eqref{f:Osim2};
    \item the number of boolean operations required to compute
      $w,v_2,\dots,v_n \in \Q[x_1]$, defined in Lemma \ref{lm:parame} as in
      \eqref{f:compv};
    \item the number of boolean operations required to compute an SOS
      decomposition of $h$ by using Algorithm $\sosone$ as in \eqref{f:bitsos1};
    \item the number of boolean operations required to compute
      $\phi_1,\dots,\phi_n$ in the output of $\sosshape$ by using
      Algorithm $\mathtt{Eliminate}$  (mentioned in
      Proposition~\ref{prop:bitshape}).
    \end{enumerate}
This sum equals
      \begin{align*}
        \widetilde
        O\Big(n^2(d+\tau)d^{2n+1}\binom{n+d}{d}+(\tau+n+d)^2d^{6n}
        +(\tau+n+d)d^{3n+1}\Big(\frac{d^{n+1}}{2}\Big
        )^{3d^{n+1}/2}+\\ (\tau+n+d)n^2d^{3n+2}\Big).
      \end{align*}
      In this sum, the third term is larger than the first and last terms for
      large enough $d$ and $n$, yielding the estimate \eqref{f:bitAlg41a}.

If in Step 3 we use $\sostwo$, the number of
	boolean operations of the algorithm is
	$$\widetilde O\Big(n^2(d+\tau)d^{2n+1}\binom{n+d}{d}+(\tau+n+d)^2d^{6n}+
	(\tau+n+d)d^{6n+4}	+n^2(\tau+n+d)d^{3n+2}\Big).$$
	Noting that $\binom{n+d}{d}\leq  (d+1)^n\leq d^{2n}$ for large enough $d$
	and $n$, we obtain \eqref{f:bitAlg41b}.
	\end{proof}

	\begin{theorem}
        Assume that $f\in \Q[\x]$ satisfies the conditions of
        Theorem \ref{thm:bit}. Let $w,v_2,\dots,v_n$, $h$ be as in Proposition
         \ref{prop:bitsos1}.
		Then, the
		maximum bitsize of the coefficients involved in the SOS decomposition
		of $h$
		obtained by using Algorithm $\sosone$  and
		Algorithm $\sostwo$  are bounded
		from
		above, respectively, by
		\begin{equation}\label{f:bitsize1}
			O\Big((\tau+n+d)( d^{n+1}/2)^{3d^{n+1}/2}d^{3n+1}\Big),
		\end{equation}
		and
		\begin{equation}\label{f:bitsize2}O\left((\tau+n+d)d^{5n+2}\right).
        \end{equation}
	\end{theorem}

	\begin{proof}
		From the proof of Proposition \ref{prop:bitsos1}, the estimates for
		degree and
		bitsize of $\beta h$
		are $d(d-1)^n$ and $\footnotesize{\widetilde
			O}\left((\tau+n+d)d^{3n+1}\right)$, respectively. According to
			\cite[Theorem~16]{univsos} and \cite[Theorem 23]{univsos}, the
			maximum
			bitsize
		of the
		coefficients involved in the SOS decomposition of $\beta h$ obtained by
		using
		$\sosone$  and  $\sostwo$ are bounded from above by
		\eqref{f:bitsize1} and \eqref{f:bitsize2}, respectively.
	\end{proof}

	\section{SOS of rational fractions modulo gradient ideals}\label{sec:rat}
	Artin's Theorem \cite{artin1927} states that if $f\in \R[\x]$ is
	non-negative then there exists a nonzero $g\in \R[\x]$ such that $g^2f$ is
	SOS, yielding a decomposition of $f$ as an SOS of rational fractions. In
	this section, we explain how to decompose $f \in \Q[\x]$ as an SOS of
	rational fractions modulo its gradient ideal.
	One says that $f \in \Q[\x]$ is an \textit{SOS of rational fractions}
	in $\Q(\x)$, where $\Q(\x)$ is the field of rational fractions in the
	variable $\x$ over $\Q$, if there exist rational fractions
	$f_1,\dots,f_s$ in $\Q(\x)$ and $[c_1,\dots,c_s] \subset \Q_+$ such that
	$f=\sum_{j=1}^sc_jf_{j}^{2}$. Furthermore, $f$ is an SOS of
	rational fractions over the quotient ring $\Q(\x)/\I_{\grad}(f)$ if
	there exists $g\in \I_{\grad}(f)$ such that $f-g$ is an SOS of
	rational fractions in $\Q(\x)$, i.e., $f$ can be decomposed as
	follows:
	\begin{equation*}\label{f1:SOSf}
		f=\sum_{j=1}^sc_jf_{j}^{2}+
		\sum_{i=1}^{n}\phi_i\frac{\partial f}{\partial x_i},
	\end{equation*}
	for some rational fractions
    $f_1, \dots, f_s, \phi_1, \dots, \phi_s$ in $\Q(\x)$ and $[c_1,\dots,c_s] \subset \Q_+$.

	\subsection{The existence of an SOS decomposition over the rationals}
	Denote by $\Q(x_1)[x_2,\dots,x_n]$  the vector space of
	polynomials in $n-1$ variables $(x_2,\dots,x_n)$ with coefficients in
	$\Q(x_1)$.

	In the following theorem, we prove the existence of an SOS decomposition of
	rational fractions modulo the gradient ideal for $f$.

	\begin{theorem}\label{thm:SOSF}
		Assume that $f\in \Q[\x]$ is a non-negative polynomial of degree $d$ and that
		$\I_{\grad}(f)$ is zero-dimensional and radical. Let $\mathcal{Q} =
		((w,\kappa_1,\dots,\kappa_n),x_1)$ be a zero-dimensional rational
		parametrization of
		$V_{\grad}(f)$. Then,
         $f$ can be decomposed as an SOS of rational
		fractions modulo the gradient ideal, namely
		\begin{equation}\label{f:SOSR}
f=\frac{1}{(w')^d}
\sum_{j=1}^sc_jq_{j}^{2}+\sum_{i=1}^{n}\phi_i(x_i-\frac{\kappa_i}{w'}),
		\end{equation}
		for some rational fractions $q_1, \dots, q_s\in\Q(x_1)$, $\phi_1,
		\dots, \phi_n \in \Q(x_1)[x_2,\dots,x_n]$, and $[c_1,\dots,c_s] \subset \Q_+$.
	\end{theorem}

	\begin{proof}  By substituting $x_i=\kappa_i/w'$ in $f$, for
		$i=2,\dots,n$, one has
		\begin{equation}\label{f:barh}
			f\Big( x_1,\frac{\kappa_2}{w'},\dots,\frac{\kappa_n}{w'}\Big) =
			\frac{1}{(w')^d}\bar h,
		\end{equation}
		where $\bar h(x_1)$ is a univariate polynomial.
		Since $f$ is non-negative with even degree $d$, $\bar h$ is also
		non-negative.
		In addition, the coefficients of $w',\kappa_1,\dots,\kappa_n$ and $f$
		are rational numbers, so the coefficients of $\bar h$ are also
		rational numbers. Applying Theorem \ref{w_SOS} for $\bar h$, we
		conclude that there are  $q_1, \dots,q_s \in
		\Q[x_1]$ and $[c_1,\dots,c_s] \subset \Q_+$ such that
		\begin{equation}\label{f:bhsos}
			\bar h=\sum_{j=1}^{s}c_jq_j^2.
		\end{equation}

		Next, one considers the division of $(w')^df-\bar h$ by
		$[w'x_1-\kappa_1,\dots,w'x_n-\kappa_n]$ with the lexicographic order
		$x_1<\dots<x_n$. Based on Buchberger's Criterion \cite{buchberger1965},
		we
		can show that this system is a Gr\"{o}bner basis of the ideal generated
		by
		this system w.r.t the order $<$ in $\Q[\x]$. Hence, there exist a
		(unique)
		list of quotients $u_1,\dots,u_n$ in $\Q[\x]$, and $r$ in $\Q[x_1]$ such
		that
		\begin{equation}\label{f:fhr}
			(w')^df-\bar h=\sum_{i=1}^{n}u_i(w'x_i-\kappa_i)+r,
		\end{equation}
		with $r$ of smaller degree than the cardinality $\delta$ of
       $V_{\grad}(f)$.
		The gradient variety of $f$ can be represented as follows:
		$$V_{\grad}(f)=\{\x\in \C^n: w=0, w'x_1-\kappa_1=\dots
		=w'x_n-\kappa_n=0\}.$$
		From \eqref{f:barh}, one sees that  $(w')^df-\bar h$ vanishes on
		$V_{\grad}(f)$.	With the same arguments as in the proof of Theorem
		\ref{thm:main}, we conclude that $r\equiv 0$. Hence, from
		$\eqref{f:barh}$,  \eqref{f:bhsos}, and \eqref{f:fhr}, we obtain a
		representation of $f$ as in \eqref{f:SOSR}, where
		$\phi_i=u_i/(w')^{d-1}$.
	\end{proof}

{In Theorem \ref{thm:4}, we assume that $\mathcal{Q} =
    ((w,\kappa_1,\dots,\kappa_n),x_1)$ is a zero-dimensional rational
    parametrization of $V_{\grad}(f)$ which is a generic assumption. In this
     assumption, the linear form  $\lambda$ is given by $\lambda(\x)=x_1$. If
      the assumption does not hold, we can change the coordinate system such
       that the obtained polynomial (with new variables) satisfies the
        assumption as in Case 2 of the proof of Theorem \ref{thm:main}.}
	\begin{remark}
		From \eqref{f:barh}, we see that $\deg \bar h$ does not exceed
		$\deg_{x_1} f+d\deg(w')$, where $\deg_{x_1} f$ is the degree of $f$ in
		the variable $x_1$ and $\deg w'=\deg w -1$. Thus, the degree of the
		univariate polynomial $\bar h$ is at most $d(d-1)^n$.
	\end{remark}

	\subsection{Algorithm to compute an SOS of rational fractions}
	From the proof of Theorem \ref{thm:SOSF}, we design an algorithm named
	$\sosgrad$ to compute the SOS decomposition of rational fractions
	for $f$. Algorithm $\sosgrad$ is obtained by  a  modification
	of Step 1 in $\sosshape$ to get a zero-dimensional
	rational parametrization of the gradient variety of $f$.

	\begin{algorithm}
		\caption{Computing SOS of rational fractions modulo the gradient
			ideal}\label{alg:SOSf}

		$\sosgrad:=\text{proc}(f)$

		\textbf{Input:} $f\in\Q[\x]$ of degree $d$ such that $f$ is
         non-negative over $\R^n$ and
		$\I_{\grad}(f)$ is zero-dimensional and radical

		\textbf{Output:}  $[w,\kappa_1,\dots,\kappa_n]$,
		$[q_1,\dots,q_s] \subset \Q[x_1]$,
		$[\phi_2, \dots, \phi_n] \subset \Q[\x]$, and $[c_1,\dots,c_s] \subset \Q_+$ satisfying
		\begin{equation}\label{eq:frat}
f=\frac{1}{(w')^d}\sum_{j=1}^sc_jq_{j}^{2}+
\sum_{i=1}^{n}\frac{\phi_i}{(w')^d}\Big(x_i-\frac{\kappa_i}{w'}\Big).
		\end{equation}

		\begin{itemize}
			\item [\rm 1:] Compute a zero-dimensional rational parametrization
			$[w,\kappa_1,\dots,\kappa_n]$ of $V_{\grad}(f)$
			\item [\rm 2:] Compute the quotients $[\phi_2,\dots,\phi_n]$ and the remainder $\bar h$ of the division of $(w')^df$ by $[x_2-\frac{\kappa_2}{w'},\dots,x_n-\frac{\kappa_n}{w'}]$ by performing $ \mathtt{Eliminate}((w')^df,w',\kappa_2,\dots,\kappa_n)$
			\item [\rm 3:] Compute a rational weighted SOS decomposition of
			$\bar h=c_1q^2_1+\dots+c_sq_s^2$
			\item [\rm 4:] Return $[w,\kappa_1,\dots,\kappa_n]$,
			$[q_1,\dots,q_s]$, $[\phi_2,\dots,\phi_n]$, and $[c_1,\dots,c_s]$
		\end{itemize}
	\end{algorithm}

	The input of $\sosgrad$ is a non-negative polynomial $f$ in
	$\Q[\x]$ whose gradient ideal $\I_{\grad}(f)$ is zero-dimensional.
	The outputs are a zero-dimensional rational parametrization of
	$V_{\grad}(f)$, a list of polynomials $[q_1,\dots,q_s] \subset \Q[x_1]$,
	and a list of rational fractions $[\phi_2,\dots,\phi_n] \subset
	\Q(x_1)[x_2,\dots,x_n]$ satisfying \eqref{eq:frat}.

	In Step 1, we compute a zero-dimensional rational parametrization
	$[w,\kappa_1,\dots,\kappa_n]$ of $V_{\grad}(f)$.
	In Step 2, we compute the quotients $[\phi_2,\dots,\phi_n]$ of the
	division of $(w')^df$ by
	$[x_2-\frac{\kappa_2}{w'},\dots,x_n-\frac{\kappa_n}{w'}]$ while using
	Algorithm $\mathtt{Eliminate}$. Note that the remainder of this division coincides with $\bar h$ given in \eqref{f:barh}.
	In Step 3, we compute a rational weighted SOS decomposition of the univariate polynomial $\bar h$ by
	relying on $\sosone$  or  $\sostwo$.

	The correctness of $\sosgrad$ is proved in a similar way as for
	$\sosshape$ in Theorem~\ref{thm:corectness}.

	\begin{theorem}
		Suppose that $f \in \Q[\x]$ is  non-negative over $\R^n$ and
		$\I_{\grad}(f)$ is zero-dimensional and radical.
		On input $f$, Algorithm $\sosgrad$ terminates and the
		outputs provide us an SOS decomposition of $f$ as in \eqref{eq:frat}.
	\end{theorem}

	\subsection{Bit complexity analysis}
 We now	estimate the bitsizes of polynomials in the output as
	well as the number of boolean operations required to perform Algorithm
	$\sosgrad$.

	\begin{proposition}\label{prop:bitquotientsrat} Assume that $\tau$ is the
         maximum bitsize of the
        coefficients of $f$ in the input of $\sosgrad$. To compute the list
		$[\phi_2,\dots,\phi_n]$ in the output,
		Algorithm $\mathtt{Eliminate}$  runs in
		$\footnotesize{\widetilde O}\left(n^2(\tau+n+d)d^{n+1}\right)$ boolean
		operations.
		Furthermore,
		the bitsize  of $\phi_i$ is $\widetilde
		O\left(n(\tau+n+d)d^{n+1}\right)$, $i=2,\dots,n$.
	\end{proposition}

	\begin{proof} We compute the division of $(w')^df$ by
	$[x_2-\frac{\kappa_2}{w'},\dots,x_n-\frac{\kappa_n}{w'}]$ by performing
     $\mathtt{Eliminate}((w')^df,w',\kappa_2, \dots,\kappa_n)$. We obtain the
      list
of quotients $[\phi_2,\dots,\phi_n]$ {and the remainder $\bar h$.}
   \if we are
 interested in only the  quotients. Note that when we perform
  $\mathtt{Eliminate}((w')^df-\bar h,w',\kappa_2, \dots,\kappa_n)$ and
   $\mathtt{Eliminate}((w')^df,w',\kappa_2, \dots,\kappa_n)$, we obtain the
    same quotients in the outputs. Hence, to obtain $[\phi_2,\dots,\phi_n]$, \fi
The degree of $(w')^df$ in
          $x_2,\dots,x_n$ is $d$, and  $\hht((w')^df)=\footnotesize{\widetilde
              O}\left((\tau+n+d)d^{n+1}\right)$. The
		conclusions are obtained by applying Lemma~\ref{lm:divshape} with
		$\hht(\kappa_i)=\footnotesize{\widetilde
O}\left((\tau+n+d)(d-1)^{n}\right)$.
	\end{proof}

	\begin{theorem}\label{thm:4}
{Let $f\in \Q[\x]$ of degree  $d$ and let $\tau$ be the maximum bitsize of
     its
 coefficients. Assume that $f$ is non-negative over $\R^n$ and
$\I_{\grad}(f)$ is zero-dimensional and radical. Then, on input $f$, Algorithm
$\sosgrad$ uses}
		\begin{equation}\label{f:bit}
\widetilde O\Big((d^{n+1}/2)^{3d^{n+1}/2}(\tau+n+d)d^{n+1}\Big),
		\end{equation}
		or
		\begin{equation}\label{f:bit1}
\widetilde O((\tau+n+d)d^{4n+4})
		\end{equation}
		boolean operations if in Step 3 we use Algorithm $\sosone$ or
		Algorithm $\sostwo$, respectively.
	\end{theorem}
	\begin{proof}
{From Corollary \ref{cor:bit1}, the
    polynomials $w,\kappa_1,\dots,\kappa_n$ in the zero-dimensional
     parametrization of the gradient variety $V_{\grad}(f)$ have degree at most
    $(d-1)^n$ and
    bitsize $ \footnotesize{\widetilde O}\left((\tau+n+d)(d-1)^{n}\right)$. We
     can see that the degree of the remainder $\bar h$ (as defined
    in \eqref{f:barh}) in Step 2 of $\sosgrad$ is at most $d(d-1)^{n}+d$
    and its bitsize is $\footnotesize{\widetilde
         O}\left((\tau+n+d)d^{n+1}\right)$.
      To compute an SOS decomposition of $\bar h$, by applying
      \cite[Theorem 17]{univsos} and \cite[Theorem~24]{univsos}, Algorithm
$\sosone$ and Algorithm $\sostwo$ use
\begin{equation}\label{f:bitSOSF1}
    \widetilde O\Big((
    d^{n+1}/2)^{3d^{n+1}/2}(\tau+n+d)d^{n+1}\Big)
\end{equation}
and
\begin{equation}\label{f:bitSOSF3}
    \widetilde O\left((\tau+n+d)d^{4n+4}\right)
\end{equation}
boolean operations, respectively.}

The estimates \eqref{f:bit} and \eqref{f:bit1} are obtained from Corollary
 \ref{cor:bit1}, Proposition \ref{prop:bitquotientsrat}, and the estimates
  \eqref{f:bitSOSF1} and \eqref{f:bitSOSF3}
		with the same line of reasoning as in the proof of Theorem
		\ref{thm:bit}.
		\if{
			Suppose that in Step 4 we use  $\sosone$ to compute an SOS
			decomposition
			of $\bar h$. Then, the number of boolean operations that
			$\sosgrad$ uses to compute the SOS decomposition of $f$ is
			the sum
			of: the number of boolean operations of the algorithm to computethe
			zero-dimensional rational parametrization $\mathcal{Q}$ of
			$V_{\grad}(f)$ as
			in \eqref{f:Osim2},  that to compute an SOS decomposition of $\bar
			h$
			by using
			$\sosone$  as in \eqref{f:bit}, and that to compute the
			list
			$\phi_1,\dots,\phi_n$ in the output of $\sosgrad$ by
			using
			Algorithm $\mathtt{Eliminate}$  (given in
			Proposition~\ref{prop:bitquotientsrat}). Repeating the argument in
			the
			proof of Theorem \ref{thm:bit}, we can reduce the sum and obtain the
			number as in \eqref{f:bit}.

			For the case using $\sosone$ in Step 4, \eqref{f:bit1} is
			proved
			similarly.}\fi
	\end{proof}

	\begin{theorem} Assume that $f\in \Q[\x]$ satisfies the conditions of
         Theorem \ref{thm:4}. Then,
the maximum bitsizes of the coefficients involved in the SOS
		decomposition of $\bar h$, obtained by using Algorithm
		$\sosone$ and
		Algorithm $\sosone$, are bounded from above respectively by
		\begin{equation}\label{f:bitSOSF2}
		{\widetilde O}\Big(( d^{n+1}/2)^{3d^{n+1}/2}(\tau+n+d)d^{n+1}\Big)
		\end{equation}
		and
		\begin{equation}\label{f:bitSOSF4}
		{\widetilde O}\left((\tau+n+d)d^{3n+3}\right).
		\end{equation}
	\end{theorem}
	\begin{proof} From the proof of Theorem \ref{thm:4}, the degree of
		$\bar h$ is at most $d(d-1)^{n}$ and the  bitsize of $\bar h$ is
		$\footnotesize{\widetilde O}\left((\tau+n+d)d^{n+1}\right)$. The
		conclusions follow from \cite[Theorem 16]{univsos} and \cite[Theorem
		23]{univsos}.
	\end{proof}

	\begin{remark}\label{rmk:compar} In general, $\sosgrad$ is
		faster than $\sosshape$ to certify  non-negativity of
		polynomials with rational coefficients.
		When relying on $\mathtt{univsos2}$, by comparing the estimates in
		\eqref{f:bitAlg41b} and \eqref{f:bit1}, we conclude that the number of
		boolean operations to run $\sosshape$ is $O(d^{2n})$
		times larger than the one of $\sosshape$.
		The underlying reason is that the maximal bitsizes of $w,v_2,\dots,v_n$
		are
		$(d-1)^{2n}$ times bigger than the ones of $\kappa_1,\dots,\kappa_n$
		that are obtained by a zero-dimensional rational parametrization of
		the gradient variety.
	\end{remark}

	To finish the section, we present an explicit SOS decomposition for the
	polynomial $f_S$ obtained from Scheiderer's polynomial given in Example
	\ref{ex:Schei}. Here, we rely on $\sosgrad$ to get the SOS
	decomposition.

	\begin{example}\label{ex:Schei2}
We first compute a zero-dimensional rational parametrization
		$\mathcal{Q}$ of the gradient variety $V_{\grad}(f_S)$:
        \begin{align*}
w & =4x_1^9+x_1^6-16x_1^5-4x_1^3-4x_1^2-1, \\
\kappa_1 & =15x_1^7-32x_1^6-9x_1^4-36x_1^3-6x_1-4, \\
\kappa_2 & =-3x_1^6+64x_1^5+24x_1^3+28x_1^2+9.
        \end{align*}

In $f_S$, by substituting $x_2=\kappa_2/w'$ as in \eqref{f:barh}, we
get the non-negative univariate polynomial {\footnotesize$\bar
h=1679616x_1^{36}+3359232x_1^{34}-559872x_1^{33}-13670208x_1^{32}+
11197440x_1^{31}-32799168x_1^{30}+7301664x_1^{29}+40124160x_1^{28}-
56581740x_1^{27}+118393488x_1^{26}-29030400x_1^{25}-11429649x_1^{24}+
91968984x_1^{23}\\-162286560x_1^{22}+52664472x_1^{21}-95470992x_1^{20}-
51948224x_1^{19}+37314854x_1^{18}-36173624x_1^{17}+\\
103156448x_1^{16}+
27660704x_1^{15}+94133752x_1^{14}+56849248x_1^{13}+51186288x_1^{12}+
42348048x_1^{11}+20765728x_1^{10}\\
+17391200x_1^9+7273168x_1^8+4607744x_1^7+
			1946186x_1^6+880960x_1^5+413632x_1^4+86580x_1^3+75816x_1^2+6561$}.

Based on Algorithm $\mathtt{Eliminate}$, we obtain the quotients of the
		division
in Step~3 of $\sosgrad$: $\phi_1=0$ and $\phi_2$ given at
\href{https://polsys.lip6.fr/~hieu/phisos.mm}{polsys.lip6.fr/$\sim$hieu/phisos.mm}.

By using $\mathtt{univsos2}$ to compute an SOS decomposition of $\bar h$,
we obtain the list $\sos$ given at above link such that
 $\bar h=\sum_{i=1}^{m}\sos[2i-1]*\sos[2i]^2$,
where $\sos[i]$ stands for the $i$-th entry of $\sos$, $m$ is the half
length of $\sos$.

Combining the above results, we obtain an SOS of rational
fractions modulo the gradient of $f_S$ as in \eqref{eq:frat}.
	\end{example}

	\section{Practical experiments}
	This section is dedicated to show experimental results obtained by using
	the algorithms $\sosshape$ (Algorithm \ref{alg:SOSdecom}
	from Section \ref{sec:pol}) and $\sosgrad$ (Algorithm
	\ref{alg:SOSf} from Section \ref{sec:rat}).
	Both algorithms are implemented in {\sc Maple}, and the results are
	obtained on
	an Intel Xeon E7-4820 CPU (2GHz) with 1.5 TB  of RAM.

	In practice, Algorithm $\sostwo$ runs faster than Algorithm
	$\sosone$, which is consistent with the theoretical results
	stated in \cite[Theorem 17]{univsos} and \cite[Theorem
	24]{univsos}.
	In addition, as mentioned in Remark \ref{rmk:compar}, it is practically
	faster to compute  SOS decompositions involving rational fractions than
	polynomials.
	We compare timings of the slowest algorithm, namely
	$\sosshape$ using $\sosone$ with the fastest
	algorithm, namely $\sosgrad$ using $\sostwo$.
	For each algorithm, the first step consists of obtaining $h$ by computing
either the shape position (using the procedure $\mathtt{Basis}$ in {\sc	Maple})
	in $\sosshape$ or the zero-dimensional rational
	parametrization (using the procedure
	$\mathtt{RationalUnivariateRepresentation}$ in {\sc Maple}) in
	$\sosgrad$. The runtime of this step is denoted by $t_h$.
	The degree and the bitsize of $h$ are denoted by $d_h$ and $\tau_h$,
	respectively.
	The second step outputs an SOS decomposition of the non-negative univariate
	polynomial $h$ by using either Algorithm $\sosone$ in
	$\sosshape$ or Algorithm $\sostwo$ in
	$\sosgrad$.
	Here, $t_{\sos}$ is the runtime of the second step and $\tau_{\sos}$ is the
	maximal bitsize of the output polynomials.

	\medskip

	\begin{center}
		{\footnotesize \begin{tabular}{ccc|c|cc|cc|c|cc|cc|}
				\cline{4-13}
				&     &   & \multicolumn{5}{c|}{$\sosshape$} &
				\multicolumn{5}{c|}{$\sosgrad$}  \\ \cline{4-13}

				&      &  &       & \multicolumn{2}{c|}{bitsize $10^6$-bits} &
				\multicolumn{2}{c|}{time (s)} &       &
				\multicolumn{2}{c|}{bitsize
					$10^4$-bits} & \multicolumn{2}{c|}{time (s)} \\ \hline
				\multicolumn{1}{|c|}{$n$} & $\tau$ & $\delta$ & $d_h$ &
				$\tau_h$& $\tau_{\sos}$
				& $t_h$ & $t_{\sos}$& $d_h$ & $\tau_h$            &
				$\tau_{\sos}$           &
				$t_h$         & $t_{\sos}$\\ 

				\multicolumn{1}{|c|}{2}   & 74   &  9 & 32    &
				0.3                 &
				8.1                & 0.1          & 2.6           & 36    &
				0.5                 & 1.6               & 0.1           & 1.8
				\\ 

				\multicolumn{1}{|c|}{3}   & 149   & 27  & 104    &
				2.4                  &
				153               & 1.1          & 781         & 108    &
				6.6
				& 13.4              & 0.2          & 13.3          \\ 

				\multicolumn{1}{|c|}{4}   &  312    & 81   & 320    &
				117                  &
				--                   & 399         & --      & 324    &
				88                &
				169                & 3.9          & 505         \\ 

				\multicolumn{1}{|c|}{5}   & 590   & 243  &    &
				&                   &   --       &        & 972    &
				940                &
				1306               & 169          & 4965         \\ \hline
			\end{tabular}

			\smallskip

			\textbf{Table 1.} Comparison results of output size and performance
			between
			Algorithm $\sosshape$ and Algorithm
			$\sosgrad$}
	\end{center}
	\smallskip

	In Table 1, we consider random polynomials of fixed degree $d=4$ with
	number of variables $n$ being between $2$ and $5$, generated as follows:
	$a^4+b_1^2+\dots+b_n^2+c+10^{6},$
	where $a$ (resp., $b_i$, $c$) is a dense linear (resp., quadratic, cubic)
	polynomial in $n$ variables. Coefficients of $a$ (resp., $b_i$, $c$) are
	chosen randomly in $\{-1,1\}$ (resp., $\{-3,\dots, 3\}$, $\{-1,0, 1\}$)
	with respect to the uniform distribution.
	For $n \geq 4$, $\sosshape$ failed to provide an SOS
	decomposition as the execution of $\sosone$ did not finish after
	12
	hours of computation, as indicated by the symbol~$-$ in the corresponding
	lines.  The underlying reason is that $\tau_h$ and $d_h$ are both very
	large
	and that the complexity of $\sosone$ is exponential in the degree
	of $h$ \cite[Theorem~17]{univsos}.
	Note that the intermediate polynomials correspond to worst cases, i.e., the
	maximal possible degree of $w$ is attained, namely $\deg w=(d-1)^n$, so the
	degree of $h$ is also maximal, i.e.,  $\deg h = d(d-1)^n-d$ (resp.
	$d(d-1)^n$) in $\sosshape$ (resp. in
	$\sosgrad$).
	For such cases, $\sosgrad$ cannot compute decompositions for
	$n\geq
	4$ (corresponding to $\deg h \geq 324$) within 12 hours.

	Next, we compare the performance of $\sosgrad$ (using
	$\sostwo$) and Algorithm $\mathtt{multivsos}$
	\cite{magron2018}.
	Recall that $\mathtt{multivsos}$ is designed to compute SOS decompositions
	of
	polynomials lying in the interior of the SOS cone.
	We report our experimental results in Table 2, obtained with seven classes
	of
	50 randomly generated polynomials.
	The random polynomials corresponding to the four first rows, with $d=4$ and
	$n=2,\dots,5$, are obtained a similar way: $a^4+b_1^2+b_2^2+c+10^{6},$
	where $a$ (resp., $b_i$, $c$) is a dense linear (resp., quadratic, cubic)
	polynomial in $n$ variables. Coefficients of $a$ (resp., $b_i$, $c$) are
	chosen randomly in $\{\pm 1,\pm 2\}$ (resp., $\{-3,\dots, 3\}$, $\{-1,\dots,
	1\}$)  with respect to the uniform distribution.
	The polynomials from the three last rows, with $d=6$ and $n=2,3,4$, are
	constructed in a similar way: $a^6+b^2+c+10^{6},$
	where $a$ (resp., $b$, $c$) is a dense linear (resp., cubic, cubic)
	polynomial
	in $n$ variables. Coefficients of $a$ (resp., $b_i$, $c$) are chosen
	randomly
	in $\{\pm 1,\pm 2\}$ (resp., $\{-3,\dots, 3\}$, $\{-1,\dots, 1\}$) with
	respect to the uniform distribution.
	Note that here the univariate polynomials generated when running the
	algorithm
	do not correspond to the worst case scenario in terms of degree and
	bitsize.
	For both algorithms, we denote by $\tau$ ($10^4$-bits) the average bitsize
	of
	the output and by $t$ the average runtime in seconds.

	\medskip
	\begin{center}

		{\footnotesize\begin{tabular}{c|ccc|c c|}
				\cline{2-6}
				&\multicolumn{3}{c|}{$\mathtt{multivsos}$}  &
				\multicolumn{2}{c|}{$\sosgrad$} \\ \hline
				\multicolumn{1}{|c|}{$d,n$} & success &  $\tau$  & $t$ &
				$\tau$      & $t$  \\ 
				\multicolumn{1}{|c|}{4,2}  & 100$\%$               & 1.3 &
				0.16      & 2     & 2       \\ 
				\multicolumn{1}{|c|}{4,3}  & 94$\%$               & 3.7 &
				0.26      & 18    & 22        \\ 
				\multicolumn{1}{|c|}{4,4}  & 38$\%$                & 8.9 &
				0.18      &   78  &   153       \\ 
				\multicolumn{1}{|c|}{4,5}  & 8$\%$                & 12.5 &
				0.32     & 234    &  630     \\ 
				\multicolumn{1}{|c|}{6,2}  & 82$\%$               & 3.5 &
				0.24      & 45    & 142          \\ 
				\multicolumn{1}{|c|}{6,3}  & 0$\%$        &
				&          & 160   &   500	  \\
				\multicolumn{1}{|c|}{6,4}  & 0$\%$
				&               &          & 744    &    4662   \\ \hline
			\end{tabular}

			\smallskip
			\textbf{Table 2.}
			Comparison of performance between Algorithm
			$\sosshape$ and Algorithm~$\mathtt{multivsos}$}
	\end{center}
	\smallskip

	From this table, we deduce that when the number of variables $n$ increases,
	then the rate of success of $\mathtt{multivsos}$ decreases. {This fact is
     compatible
     with  Blekherman's theorem \cite{blekherman2006} which says that if the
      degree $d\geq 4$ is fixed then, as the number of variables $n$ grows, the
       cone of non-negative polynomials is significantly bigger than the cone  of
        SOS polynomials.}
	When $\mathtt{multivsos}$ succeeds in computing SOS decompositions, then it
	provides more concise certificates than $\sosgrad$ while being
	more efficient.
	However, when $d=4$ and $n=5$, $\mathtt{multivsos}$ can only decompose four
	polynomials out of 50 while $\sosgrad$ succeeds for all of
	them.
	This demonstrates the need of alternative procedures such as
	$\sosgrad$ for polynomials which presumably do not lie in the
	interior of the SOS cone.

	\section*{Conclusions and perspectives}
	We designed and analyzed two algorithms to  decompose a non-negative
	polynomial
	as an SOS of polynomials/rational fractions modulo the gradient ideal with
	rational coefficients.
	The correctness of our framework relies on a generic condition, namely that
	the
	gradient ideal of the input polynomial is zero-dimensional and radical.
	We shall improve the scalability of our algorithms by exploiting the
	specific
	structure of the input polynomial, such as correlative \cite{lasserre2006}
	or
	term sparsity \cite{wang2021}, symmetries \cite{riener2018} or by using
	recent improvements on the computation of critical sets when the related
	system is invariant under group actions \cite{faugere2020}.
	Furthermore, we also plan to extend our algorithms to the constrained case
	by
	relying on polar varieties as in \cite{greuet2012}.
    \\

\section*{Acknowledgements}
This work has been supported by European Union's Horizon 2020 research and
innovation programme under the Marie Sk\l{}odowska--Curie Actions, grant
agreement 813211 (POEMA). \\


\appendix
\section*{Appendix}

\section{Proof of Corollary \ref{cor:bit1}}\label{app:bitparam}
Assume that the system of partial
derivatives $\frac{\partial f}{\partial x_1} \,,\dots, \,\frac{\partial
    f}{\partial x_n}$ is given by a straight-line program $\Gamma$ of size
$L$,
i.e., the program uses $L$ elementary operations $+, -,\times$ to evaluate
the system from variables $x_1,\dots,x_n$ and integers with bitsizes at most
$\max_{i=1}^n\{\hht(\frac{\partial f}{\partial x_i})\}$.

We claim that $L$ is $O(d\binom{n+d}{d})$. Indeed,
$f$ has at most $\binom{n+d}{d}$ terms and each term in
$f$ is defined by at most $d+1$ multiplications.  Hence, the size of a
straight-line program $\Gamma_f$ which defines $f$ does not exceed
$(d+1)\binom{n+d}{d}$. By applying Baur-Strassen Theorem \cite[Theorem
1]{baur1983}, the size $L$ is $O(d\binom{n+d}{d})$.

Recall that  $\hht(\frac{\partial f}{\partial x_i}) \leq \log
d+\hht(f)=\log d+\tau$, for $i=1,\dots,n$.
By applying \cite[Corollary 2]{safey2018} for the system and a single
group of variables, there exists an algorithm that takes the system as
in input, and that produces one of the outputs given as in items \rm{a)}--
\rm{c)} of Corollary \ref{cor:bit1}. The number of boolean operations of
the
algorithm is $\footnotesize{\widetilde O}\Big(n^2d^{2n}(\log d+\tau+(d-1))
(d\binom{n+d}{d}+n(d-1)+n^2)\Big)$.
Reduce this formula, we get \eqref{f:Osim2}. Furthermore,
the polynomials in the output have degree at most $(d-1)^n$ and
bitsize $\footnotesize{\widetilde O}\left((d-1)^{n}(\log
d+\tau+n+(d-1))\right)= \footnotesize{\widetilde
    O}\left((d+\tau+n)d^{n}\right)$
as claimed.

\section{{Proof of the bit complexity
        in Lemma} \ref{lm:parame}}\label{app:q}
From Corollary~\ref{cor:bit1},  the degree of $w$ is at most
$(d-1)^n$, and
then $\deg w'$ is at most $(d-1)^n-1$. Assume that $\beta$ is the
positive
minimal common denominator of all non-zero coefficients of $w$. Then,
$\beta w$ and $\beta w'$ belong to $\Z[t]$. Clearly,
$\deg(\beta w')=\deg(\beta w)-1$, $\deg(\beta w)\leq (d-1)^n$, and the
bitsize
of $\beta w$ and $\beta w'$ are bounded by
$\footnotesize{\widetilde O}\left((d+\tau+n)(d-1)^{n}\right)$. We can
apply \cite[Theorem 6.52]{gathen2013} to $\beta w$ and $\beta w'$. The
extended Euclidean
algorithm computes the B\'{e}zout coefficient, denoted by $b$, of $ \beta
w'$ using
\begin{equation}\label{Obezout}
    \widetilde O(\tau+n+d)^2(d-1)^{6n})
\end{equation}
boolean operations. The bitzise of $b$ is bounded by
\begin{equation}\label{taub}
    O\left((\tau+n+d)(d-1)^{2n}\right).
\end{equation}
Furthermore, one sees that the degree of  $b$ satisfies
\begin{equation}\label{degb}
    \deg b\leq \deg w-\deg \gcd(w,w')= \deg w\leq (d-1)^n.
\end{equation}

For every $i=2,\dots,n$, we will estimate the bitsize of the
polynomial
$b\kappa_i$. Recall from Corollary \ref{cor:bit1} that
$\deg \kappa_i \leq (d-1)^{n}$, hence from \eqref{degb} one has $\deg
b\kappa_i \leq 2(d-1)^{n}$. From \eqref{taub}, we obtain
$$\hht(b\kappa_i)\leq \hht(b)+\hht(\kappa_i)= \footnotesize{\widetilde
    O}((\tau+n+d)(d-1)^{2n})+\footnotesize{\widetilde
    O}((\tau+n+d)(d-1)^{n}).$$
After simplifying the last estimate, the bitsize of $b\kappa_i$ is
bounded from
above by $\footnotesize{\widetilde O}((\tau+n+d)(d-1)^{2n})$. Hence,
the
bitsize of $\eta b\kappa_i$, where $\eta $ is the minimal common
denominator
of all non-zero coefficients of $b\kappa_i$, can be estimated as
follows
$$\hht(\eta b\kappa_i)\leq 2\hht(b\kappa_i) \leq
\footnotesize{\widetilde
    O}((\tau+n+d)(d-1)^{2n}).$$

In the proof of Lemma \ref{lm:parame}, we considered the division of
$b\kappa_i$ by $w$ and defined $v_i= b\kappa_i \mod w$. Thus, the
degree of
$v_i$ is at most $\deg w\leq (d-1)^n$. From Lemma~\ref{lm:div1}, the
Euclidean
division algorithm computes $v_i$ using at most
\begin{equation}\label{Odiv}
    \widetilde O((\tau+n+d)(d-1)^{5n})
\end{equation}
boolean operations. Thus, the bitsize of $v_i$ is
$\footnotesize{\widetilde O}((\tau+n+d)(d-1)^{3n})$, for $i=2,\dots,n$.
Therefore, computing $[w,v_2,\dots,v_n]$ from the zero-dimensional
rational
parametrization $\mathcal{Q}$ of $V_{\grad}(f)$, requires
$$\footnotesize{\widetilde
    O}((\tau+n+d)^2(d-1)^{6n}+(n-1)(\tau+n+d)(d-1)^{5n})$$
boolean operations, as a consequence of \eqref{Obezout} and
\eqref{Odiv}.
By applying further simplification, we obtain the desired result
\eqref{f:compv}.

\bigskip

	The bit complexity results of the two division algorithms used in Lemma
	\ref{lm:div1} and Lemma~\ref{lm:divshape} are basic but we could not find
	their proofs in the literature. Here we state these two algorithms and
	prove estimates for their bit complexities.

\section{Proof of Lemma \ref{lm:div1}}\label{app:lm:div1} Assume that $a,b$ are polynomials in $\Z[t]$  with $\deg a=d\geq\deg b=m$
and that $\hht(a)$, $\hht(b)$ are bounded from above by $\tau$. We recall
the Euclidean division algorithm in Algorithm \ref{alg:Euclid}
\cite[Algorithm~2.5]{gathen2013} to compute the quotient $q$ and the
remainder $r$ of the division of $a$ by $b$, i.e., $a=qb+r$ with $\deg
r< \deg b$.

	\begin{algorithm}
		\caption{Euclidean division algorithm}\label{alg:Euclid}
		\textbf{Input:} polynomials $a,b\in \Z[t]$

		\textbf{Output:} polynomials $q,r\in \Q[t]$ such that $a=qb+r$ and
		$\deg r< \deg b$

		\begin{description}
			\item[\rm 1:] Let $q:=0$ and $r:=a$
			\item[\rm 2:] While $\deg r \geq \deg b$ do
			\begin{enumerate}
				\item[\rm 3:] Let $h:=\lc(r)/\lc(b) t^{\deg r - \deg b}$
				\item[\rm 4:] Let $q:=q+h$
				\item[\rm 5:] Let $r:=r-hb$
			\end{enumerate}

			\item[\rm 6:] Return $q$ and $r$
		\end{description}
	\end{algorithm}



	We denote by $r_i$ (resp. $q_i$, $h_i$) the value of $r$ (resp. $q$, $h$) at
	the $i$-th iteration of the while loop from Step 2.
	The initial values are $q_1=0$ and $r_1=a$. After each iteration of the
	while loop, the degree of $r$ is strictly decreasing. Hence, the while
	loop will terminate after at most $d-m$ iterations.
	When the while loop terminates, the values are $q=q_{d-m}$ and $r=r_{d-m}$.
	From Step 3, we observe that $\hht(h_i)\leq \hht(b)+\hht(r_i) =
	\tau + \hht(r_i)$, and the number of boolean operations to perform this
	operation is bounded by $\tau+\hht(r_i)$. Since $q_{i+1}=q_{i}+h_i$ and
	$h_i$ is a monomial satisfying $\deg h_i> \deg q_{i}$, one has
	\begin{equation}\label{f:tauqmax}
		\hht(q_{i+1})\leq \max\{\hht(q_{i}),\tau+\hht(r_i)\} \,,
	\end{equation}
	and the number of boolean operations to perform the operation in Step 4 is
	bounded by $O(1)$.
	For the operation in Step 5, since $\hht(hb)\leq 2\tau+\hht(r_i)$, the
	bitsize of $r_{i+1}$ is bounded by $2\tau+\hht(r_i)$. Moreover, the number
	of boolean operations to compute $h_ib$ is $\footnotesize{\widetilde
		O}(m(\tau+\hht(r_i)))$, so $r_{i+1}$ is also computed in
	$\footnotesize{\widetilde O}(m(\tau+\hht(r_i)))$ boolean operations.

	We get the recurrence formula $\hht(r_{i+1}) \leq \hht(r_{i}) + 2 \tau$,
	for each $i=0,\dots,d-m$, with $\hht(r_0)=\tau$.
	It follows that $\hht(r_i) \leq 2i\tau+\tau$, for each $i=0,\dots,d-m$.
	This yields $$\hht(r)=\hht(r_{d-m}) \leq 2(d-m)\tau + \tau=O(\hht(d-m)).$$
	From \eqref{f:tauqmax}, the bitsize of $q=q_{d-m}$ is also bounded by
	$O(\hht(d-m))$. Furthermore, the number of boolean operations to perform
	the algorithm is
	$$\sum_{i=0}^{d-m} \footnotesize{\widetilde O}((i+1)2m\tau)
	=\footnotesize{\widetilde O}(m\hht(d-m)^2).$$
	This yields the desired estimates.

	\section{{Algorithm $\mathtt{Eliminate}$ and the proof of Lemma \ref{lm:divshape}}}\label{app:divshape}

    \noindent {\sc {Algorithm $\mathtt{Eliminate}$}}.
    Let us consider $g \in \Q[x_1][x_2,\dots,x_n]$, with $\deg g=d$ (in variables $x_2,\dots,x_n$) and $\hht(g)=\tau_g$, and the list of rational fractions:
	$$ G = [x_2-\frac{a_2}{a_0},\dots,x_n-\frac{a_n}{a_0}],$$
	where $a_0, a_2,\dots,a_n$ are polynomials in $\Q[x_1]$, $a_0\neq 0$, and $\hht(a_i)\leq \tau_a$
    for $i=0,2,\dots,n$.
	Recall that $\Q(x_1)$ is the field of
	rational fractions in variable $x_1$ with coefficients in  $\Q$.
    Let $x_2<\dots<x_n$ be a lexicographic monomial order
	 on $\Q(x_1)[x_2,\dots,x_n]$.
	Algorithm $\mathtt{Eliminate}$ outputs the  quotients
    $\phi_2,\dots,\phi_n \in
	\Q(x_1)[x_2,\dots,x_n]$ {and the remainder $r\in \Q(x_1)$} of the
	multivariate division of $g$
	by the list $G$ satisfying \begin{equation}\label{f:gphi}
		g=\sum_{i=2}^n\phi_{i}(x_i-\frac{a_i}{a_0})+r.
	\end{equation}
\if Instead of directly operating on $\frac{1}{b},a_2,\dots,a_n$, we use $n$
	new variables $z_0,z_2,\dots,z_n$ that will be replaced respectively by
	$\frac{1}{a_0},a_2,\dots,a_n$  in Step 5. \fi

	\begin{algorithm}\caption{Elimination algorithm}\label{alg:divshap}

		$\mathtt{Eliminate}:=\text{proc}(g,a_0,a_2,\dots,a_n)$

		\textbf{Input:} $n+1$ polynomials $g \in \Q[x_1][x_2,\dots,x_n]$, $a_0,a_2,
		\dots,a_n\in \Q[x_1]$

		\textbf{Output:} $\phi_2,\dots,\phi_n$ in
		$\Q(x_1)[x_2,\dots,x_n]$ {and $r\in \Q(x_1)$} satisfying \eqref{f:gphi}

		\begin{description}
			\item [\rm 1:] Set $r_{n+1}:=g$
			\item [\rm 2:] For $i=n$ to $2$ do
			\begin{description}
				\item[\rm 3:] Compute
                 $\phi_i:=\quo(r_i,x_i-\frac{a_i}{a_0},x_i)$
				\item[\rm 4:] Substitute $x_i$ by $\frac{a_i}{a_0}$ in
                 $r_{i+1}$ to define
                  $r_{i}:=r_{i+1}(x_1,\dots,x_{i-1},\frac{a_i}{a_0})$
			\end{description}
            \item [\rm 5:] {Set $r:=r_2$}
			\item [\rm 6:] {Return $\phi_2,\dots,\phi_n$, and $r$}
		\end{description}
	\end{algorithm}

	In Step 3,  $\phi_i$ is the quotient of the univariate division (in the
	variable $x_i$) of $r_i$ by $x_i-\frac{a_i}{a_0}$. Since the degree of
     $x_i$ in $x_i-\frac{a_i}{a_0}$ is 1, $\phi_i$ belongs to
      $\Q(x_1)[x_2,\dots,x_i]$.
	The remainder $r_i$ of the division in Step 3 is given in Step 4 after
	replacing $x_i$ by $\frac{a_i}{a_0}$ in $r_{i+1}$; hence one has $r_i\in \Q(x_1)[x_2,\dots,x_{i-1}]$.  After Steps
	3-4, we obtain
	\begin{equation}\label{eq:ri}
		r_i=\phi_i\Big(x_i-\frac{a_i}{a_0}\Big)+r_{i-1}.
	\end{equation}
	Therefore, after Step 5, we get
	$g=\sum_{i=2}^n\phi_{i}(x_i-\frac{a_i}{a_0})+r$, with  $r\in \Q(x_1)$.
	Based on Buchberger's Criterion \cite{buchberger1965}, we can show that the
	system of $n-1$ polynomials
	$[x_2-\frac{a_2}{a_0},\dots,x_n-\frac{a_n}{a_0}]$  is a Gr\"{o}bner basis
	of the ideal generated by this system w.r.t.~the order $<_{lex}$ in
	$\Q(x_1)[x_2,\dots,x_n]$. Hence, $\phi_2,\dots,\phi_n$
	are defined uniquely. The correctness of the
	algorithm is proved. \\

\noindent {\sc {The proof of Lemma \ref{lm:divshape}}}.	Now we estimate the bitsizes of $\phi_i$, for $i=2,\dots,n$.
	From the definition of $r_i$ in Step~4, one sees that
    $\hht(r_{i})\leq \hht(r_{i+1})+2d\tau_a$.
	Since  $\hht(r_{n+1})=\tau_g$, the bitsize
	of $r_{i}$ is bounded from above by $\tau_g+2(n-1)d\tau_a$. The relation
	\eqref{eq:ri} leads to
	$\hht(\phi_i)\leq \hht(r_{i+1}-r_{i})+\hht(x_i-\frac{a_i}{a_0})$. Because of
	$\hht(r_{i+1}-r_{i})\leq \max\{\hht(r_{i+1}),\hht(r_{i})\}$, and
	$\hht(\frac{a_i}{a_0}) \leq 2\tau_a$, we get $ \hht(\phi_i)\leq
	\tau_g+2(nd-d+1)\tau_a$. It follows that
	$\hht(\phi_i)=\footnotesize{\widetilde O}(\tau_g+nd\tau_a)$.

	We see 	that the number of boolean operations to perform Steps 3 and 4 are
	$\footnotesize{\widetilde O}(\tau_g+nd\tau_a)$ and $O(1)$, respectively.
     The for loop in Step
	2 has $n-1$ steps. Therefore, the number of boolean operations to
	perform the loop is $\footnotesize{\widetilde
		O}(n\tau_g+n^2d\tau_a)$.  This is also the number of boolean
operations that Algorithm $\mathtt{Eliminate}$ uses.


\end{document}